\documentclass[12pt]{article}

\usepackage{amssymb,amsmath,amsthm}
\usepackage{amsmath}
\usepackage{amssymb}  
\usepackage{bm}
\usepackage[mathscr]{euscript}

\usepackage{amsfonts}
  \newcommand{\field}[1]{\mathbb{#1}}

\renewcommand{\vec}[1]{\bm{#1}}

\newcommand{\mon}{\mu}
\newcommand{\vol}{\mathrm{vol}}
\newcommand{\bin}{\mathrm{Bin}}

\newcommand{\E}{\mathbb{E}}
\newcommand{\Prob}{\mathbb{P}}
\def\nfrac#1#2{{\textstyle\frac{#1}{#2}}}

\newtheorem{theorem}{Theorem}[section]
\newtheorem{lemma}[theorem]{Lemma}
\newtheorem{proposition}[theorem]{Proposition}
\newtheorem{corollary}[theorem]{Corollary}

\date{ }

\begin{document}
\title{Mixing of the Glauber dynamics for the ferromagnetic Potts model}

\author{
Magnus Bordewich\thanks{Supported by EPSRC grant EP/G066604/1}\\
\small School of Engineering and Computing Sciences\\[-0.8ex]
\small Durham University\\[-0.8ex]
\small Durham, DH1 3LE, UK\\[-0.8ex]
\small \tt m.j.r.bordewich@durham.ac.uk\\
\and
Catherine Greenhill\thanks{Research supported by the Australian Research
Council and performed while the second author was on sabbatical at the University
of Durham.  }\\
\small School of Mathematics and Statistics\\[-0.8ex]
\small The University of New South Wales\\[-0.8ex]
\small Sydney NSW 2052, Australia\\[-0.8ex]
\small \tt csg@unsw.edu.au\\
\and
Viresh Patel\thanks{Research performed while the third author was at the University of Durham.  Supported by 
EPSRC Grant EP/G066604/1}\\
\small School of Mathematical Sciences\\[-0.8ex]
\small Queen Mary, University of London\\[-0.8ex]
\small London, E1 4NS\\[-0.8ex]
\small\tt viresh.s.patel@gmail.com 
}\maketitle

\begin{abstract} 
We present several results on the mixing time of the Glauber dynamics for sampling from the Gibbs distribution
in the ferromagnetic Potts model.  At a fixed temperature and interaction strength, we study the interplay between the maximum degree ($\Delta$) of the underlying graph and the number of colours or spins ($q$) in determining whether the dynamics mixes rapidly or not. We find a lower bound $L$ on the number of colours  such that Glauber dynamics is rapidly mixing if at least $L$ colours are used. 
 We give a  closely-matching 
 upper bound $U$ on the number of colours such that with probability that tends to 1, the Glauber dynamics mixes slowly on random $\Delta$-regular graphs when at most $U$ colours are used.
We show that our bounds can be improved if 
we restrict attention to certain types of graphs of maximum degree $\Delta$, e.g.\ toroidal grids for $\Delta=4$. 
\end{abstract}

\section{Introduction}

The Potts model was introduced in 1952~\cite{pot52} as a generalisation of the Ising model of magnetism.  The Potts model has been extensively studied not only in statistical physics, but also in computer science, mathematics  and further afield. In physics the main interest is in studying phase transitions and modelling the evolution of non-equilibrium particle systems; see~\cite{wu82} for a survey. In computer science, the Potts model is a test-bed for approximation algorithms and techniques. It has also been heavily studied in the areas of discrete mathematics and graph theory, through an equivalence to the Tutte polynomial of a graph~\cite{wel93}, and thereby links to the chromatic polynomial and many other graph invariants. The Potts model and its extensions have also appeared many times in the social sciences, for example in modelling financial markets~\cite{voi01} and voter interaction in social networks~\cite{cho07}, and in biology~\cite{gla92}.  

\medskip
\noindent
{\bf Potts Model.}\  
In graph-theoretic language, the Potts model assigns a weight to each possible colouring of a graph (not necessarily proper), and we are interested in sampling from the distribution induced by the weights. The main obstacle to sampling is that the appropriate normalisation factor, the sum of the weights  of all colourings, is hard to compute. To be precise:
for a graph $G=(V,E)$, a \emph{(spin) configuration} $\sigma$ is a function which assigns to each vertex $i$ a colour $\sigma_i \in \{1, \ldots, q \}$ (also called states or spins).  
The probability of finding the system in a given configuration $\sigma$ is given by the Gibbs distribution: 
\[\pi(\sigma) = Z^{-1} e^{\beta \sum_{(i,j)\in E} 
   J \delta(\sigma_i, \sigma_j)},
\] 
where $\delta(\sigma_i, \sigma_j)$ is the Kronecker-$\delta$ (taking value 1 $\sigma_i = \sigma_j$, and taking value 0 otherwise); $\beta = (kT)^{-1}>0$ is the inverse temperature (here $k$ is Boltzman's constant and $T$ is temperature); and $Z=Z(G,\beta,J,q)$, is the \emph{partition function} i.e.\ the appropriate normalisation factor to make this a probability distribution. The strength of the interaction between neighbouring vertices is given by the coupling constant $J$. If $J>0$ then the bias is towards having many edges with like colours at the endpoints; this is the ferromagnetic region. 
If $J<0$ then the bias is towards few edges with like colours at the endpoints: this is the anti-ferromagnetic region.

Our results 
concern only the ferromagnetic region, where $J > 0$, although we discuss some background on the antiferromagnetic region below. 
We regard $e^{\beta J}$ as a single parameter $\lambda \geq 0$, which we will call the \emph{activity}; thus $\lambda>1$ gives the ferromagnetic region and $\lambda<1$ gives the antiferromagnetic region.
 Setting $\mu(\sigma)$ to be the number of monochromatic edges in a configuration $\sigma$ (that is, $\mu(\sigma)=\sum_{(i,j)\in E} \delta(\sigma_i, \sigma_j)$), we obtain the formula
\[ Z(G,\lambda,q)=\sum_{\sigma\in [q]^V} \lambda^{\mu(\sigma)}.\]

\medskip
\noindent
{\bf Computing the partition function.}\
When $q=1$ the evaluation of the partition function is trivial. It is also trivial when $q=2, \lambda=0$, which is the antiferromagnetic Ising model at zero temperature: here the partition function counts the number of proper 2-colourings of $G$. In all other cases it is \#P-hard to compute the partition function exactly, and thus there can be no efficient algorithm (running in time polynomial in the size of the underlying graph) assuming P$\neq$NP. (Note that the related Tutte polynomial has three additional points on the real plane at which it can be efficiently evaluated~\cite{jae90}, but these do not correspond to the ferromagnetic 
Potts model at physically meaningful points, \textit{i.e.} where $q\geq 1$ and $ \lambda\geq 0$.) 
As a result of the hardness of exact evaluation, attention has been focused on approximation algorithms. The specific question is: for what classes of graphs and what ranges of $q$ and $\lambda$ is there a fully polynomial randomised approximation scheme (FPRAS) for computing the partition function? 

In the anti-ferromagnetic case, $\lambda<1$, there can be no FPRAS for the partition function unless NP=RP, except when $q=1$ (for all $\lambda$) and when $q=2$ and $\lambda=0$~\cite{gol08}.
For the ferromagnetic region, $\lambda>1$, there is only known to be an FPRAS when $q=2$ (the Ising model) for general graphs at any temperature~\cite{jer93}. There is also an FPRAS for the entire ferromagnetic region (no restriction on $q$) if we restrict the underlying graphs to the class of dense graphs (those having minimum degree $\Omega(n)$~\cite{alo95}, or having edge connectivity at least $\Omega(\log n)$~\cite{kar99}). In terms of approximation complexity, 
approximating the partition function of the ferromagnetic Potts model is 
equivalent to \#BIS, which is the problem of approximating the number of independent sets in a bipartite graph~\cite{gol12}.  
This puts it in an interesting class of approximation problems, namely,
those which are \#BIS-equivalent: no such problem is known to 
be hard, but none have been shown to exhibit an 
FPRAS~\cite{dye03b}. 

\medskip
\noindent
{\bf Glauber dynamics.}\
A standard approach to approximating the partition function is to simulate \emph{Glauber dynamics}. In {Glauber dynamics} the following process is iterated (starting from any given configuration): a random vertex updates its colour by selecting a colour according to the local Gibbs distribution induced by the current colourings of its neighbours. (This will be formalised in the next subsection.) The distribution on configurations obtained after $t$ steps of Glauber dynamics converges to an equilibrium given by the global Gibbs distribution on the whole graph, as $t$ goes to infinity. The approximation is achieved by simulating the Glauber dynamics for long enough to generate a sample that is distributed with very nearly the equilibrium distribution. This process is Markov chain Monte Carlo sampling (MCMC)~\cite{jer96}. The close link between sampling and approximate counting means that if Glauber dynamics gets sufficiently close to equilibrium in polynomial time (in the size of the graph) then there is an FPRAS for the partition function. In this case the dynamics is said to mix rapidly. 

In the ferromagnetic case, physicists' 
understanding of phase transitions
  indicate that at sufficiently high temperature (all other things being equal) Glauber dynamics will mix rapidly, whereas at sufficiently low temperature Glauber dynamics will mix slowly~\cite{mar99}. 
The intuitive explanation is as follows: at high temperature in the ferromagnetic region, $\beta$ is small and so $\lambda$ is close to 1; thus all configurations are weighted roughly equally
 and the Glauber dynamics walks freely over the state space without getting `stuck'. 
At low temperatures in the ferromagnetic region, $\beta$ is large and so $\lambda$ is also large; thus 
configurations consisting of predominantly one colour are far more heavily weighted than configurations with a balance of colours,
 so the Glauber dynamics will become trapped in configurations of the former type.
However, determining the exact  range of temperature in which Glauber dynamics mixes rapidly is, in general, open.

In the {anti-ferromagnetic} case, where it is known that there can be no FPRAS in general, the MCMC technique has still yielded many results approximating the partition function for restricted classes of graph, notably bounded-degree graphs. In the zero temperature limit of the anti-ferromagnetic Potts model only proper vertex colourings have non-zero weight. Thus approximating the partition function is equivalent to approximately counting {proper} $q$-colourings of the underlying graph. Jerrum~\cite{jer95} first showed that provided  the number of colours is more than twice the maximum degree of the graph then the Glauber dynamics will mix rapidly, also proved independently  in the physics community by Salas and Sokal~\cite{sal97}. This result has been followed by numerous refinements gradually reducing the ratio of colours to degree required for rapid mixing: see~\cite{fri07} for a recent survey. In this paper we shall investigate the interplay of the maximum degree $\Delta$ of the graph $G$ and the number of colours  $q$ in determining whether the convergence of Glauber dynamics for the \emph{ferromagnetic} Potts model is fast (rapid mixing) or slow.

\subsection{Definitions}

Throughout we shall be concerned with discrete-time, reversible, ergodic Markov chains with finite state space $\Omega$. Let 
$\mathcal{M}$ be such a Markov chain with transition matrix $P$ and (unique) stationary distribution $\pi$. For $\varepsilon > 0$ and $x \in \Omega$, we define
\[
 \tau_x(\mathcal{M}, \varepsilon) = \min \{ t: \, \|P^t(x, \cdot) - \pi(\cdot)\|_{TV} \leq \varepsilon \},
\]
where $\| \cdot \|_{TV}$ denotes total variation distance between two distributions: that is, \[
\| \phi - \phi' \|_{TV} := \frac12 \sum_{x \in \Omega}|\phi(x) - \phi'(x)|.
\]
for any two probability distributions $\phi$, $\phi'$  on $\Omega$.
We define $\tau(\mathcal{M},\varepsilon)= \max_x \tau_x(\mathcal{M}, \varepsilon)$.

Let $G=(V,E)$ be a graph with $n:=|V|$,
 and let $[q]=\{ 1, \ldots, q \}$ be a set of colours (or spins).
  We write 
$\Omega = [q]^V$ for the set of configurations  of $G$ (\textit{i.e.} not-necessarily proper $q$-colourings). Fix a constant $\lambda > 1$, which is called the \emph{activity}. 
The \emph{Gibbs distribution} $\pi = \pi(G, \lambda, q)$ on $\Omega$ is given by 
\[ \pi(\sigma) \propto \lambda^{ \mon(\sigma)}
\]
 for all $\sigma\in\Omega$, 
where $\mon(\sigma)$ denotes the number of monochromatic edges of $G$ in the configuration $\sigma$. More precisely, $\pi(\sigma) = \lambda^{ \mon(\sigma)}/Z$, 
where  $Z$ is the partition function
\[
Z = Z(G, \lambda, q) = \sum_{\sigma \in \Omega} \lambda^{\mon(\sigma)}.
\]

The \emph{Glauber dynamics} is a very simple Markov chain on $\Omega$, with stationary distribution given by the Gibbs distribution.   Given a configuration $X\in\Omega$, a vertex $v\in V$, and a colour $c \in [q]$, let $n(X,v,c)$ denote the number of neighbours of $v$ with colour $c$ in $X$. Define the probability distribution $\phi_X^v$ on $[q]$ by
\[  \phi_X^v(c) \propto \lambda^{n(X,v,c)}.\]
The transition procedure
of the Glauber dynamics from current state $X_t\in\Omega$ is as follows:
\begin{itemize} 
\item choose
a vertex $\vec{v}$ of $G$ uniformly at random; 
\item given that $\vec{v}=v$ (here $\vec{v}$ is random and $v$ is fixed), choose a colour $c \in [q]$ according to the distribution $\phi = \phi_{X_t}^v$; 
\item for each $u \in V$ let
$X_{t+1}(u) = 
\begin{cases} 
X_t(u) 	&\text{if } u \not= v, \\
c 		&\text{if } u = v.
\end{cases}$
\end{itemize}
Then $X_{t+1}$ is the new state. We write $\mathcal{M}_{\mathrm{GD}} = \mathcal{M}_{\mathrm{GD}}(G,\lambda,q)$ for the Glauber dynamics as
described above.

We say that $\mathcal{M}_{\mathrm{GD}}$ \emph{mixes rapidly} if $\tau(\mathcal{M},\varepsilon)$ is polynomial in $\log{|\Omega|}$, that is, polynomial in $n$.
If $\tau(\mathcal{M},\varepsilon)$ is exponential in $n$, then we say that $\mathcal{M}_{\mathrm{GD}}$ \emph{mixes slowly}.

\subsection{Results}
Our main results are stated below. In order to keep the presentation simple at this stage, we sometimes postpone giving the explicit relationships amongst constants and mixing times until later, but in each case, we direct the reader to where a more detailed statement can be found.

In Theorem~\ref{th:main1} we present our first, and simplest, bound on the number of colours, as a function of $\lambda$ and $\Delta$, that guarantees rapid mixing of Glauber dynamics. 
Although Theorem~\ref{th:main1} follows from a standard coupling argument, 
for completeness we prove it here, as we will need this result
later to establish our improved bounds.

\begin{theorem}
\label{th:main1}
Let $\Delta,q\geq 2$ 
be integers and take $\lambda>1$ such that 
$q\geq \Delta\lambda^\Delta + 1$.
Then the Glauber dynamics of the $q$-state Potts model at activity $\lambda$ mixes rapidly for
the class of graphs of maximum degree $\Delta$.
\end{theorem}
Theorem~\ref{th:main1} will be proved in Section~\ref{se:GD};
see Proposition~\ref{le:vertex} for a more detailed statement.

In Theorem~\ref{th:main2} we improve the exponent of $\lambda$ in the bound, but at the expense of a larger constant. We also show that the exponent achieved is close to the best possible, by proving a corresponding slow-mixing bound for almost all regular graphs of degree $\Delta$.
\begin{theorem}
\label{th:main2}
Fix an integer $\Delta \geq 2$. For any $\eta \in (0,1)$ 
there are
constants $c_1$ and $c_2$ (depending on $\eta$ and $\Delta$), such that for any integer $q \geq 2$ and any $\lambda>1$  
\begin{enumerate}
\item[{\rm (i)}] if $q>c_1\lambda^{\Delta-1+\eta}$ then the Glauber dynamics of the $q$-state Potts
model at activity $\lambda$ mixes rapidly for the class of connected 
graphs of maximum degree $\Delta$;
\item[{\rm (ii)}] if $q<c_2\lambda^{\Delta-1-\tfrac{1}{\Delta-1}-\eta}$ then the Glauber dynamics of
the $q$-state Potts model at activity $\lambda$ mixes slowly for almost all regular graphs of
degree $\Delta \geq 3$.
\end{enumerate}
\end{theorem}

Theorem~\ref{th:main2} is proved at the end of the paper: a more detailed statement of Theorem~\ref{th:main2}(i) can be found
in Theorem~\ref{th:glaubercompare}, while a more detailed statement of Theorem~\ref{th:main2}(ii) can be found in Theorem~\ref{th:torpid-random}.

Theorem~\ref{th:main2}(ii) is proved using a conductance argument. It turns out that conductance for the Glauber dynamics is related to the expansion properties of the underlying graph, and so we prove that almost all $\Delta$-regular graphs have the relevant property. This argument alone gives a worse bound than that in Theorem~\ref{th:main2}(ii), but combined with the solution of an interesting extremal problem (proved in Section~\ref{se:extremal}), which we believe may be of independent interest, we are able to obtain the required improvement. 

Theorem~\ref{th:main2}(i) is proved by first using a coupling argument to prove a rapid-mixing result for block dynamics (a more general form of dynamics than Glauber dynamics) and then using a Markov chain comparison argument to obtain rapid mixing for Glauber dynamics. In proving  Theorem~\ref{th:main2}(i), we derive a general combinatorial condition on graphs that guarantees rapid mixing of Glauber dynamics (Theorem~\ref{se:grid}  combined with Corollary~\ref{co:comparison}). This condition can be used to improve the bounds of
Theorem~\ref{th:main2}(i) 
for graph classes of maximum degree $\Delta$ with ``low expansion''. We illustrate this in Theorem~\ref{th:main3} below with the example of the toroidal grid.

\begin{theorem}
\label{th:main3}
 For any $\eta\in (0,1)$ 
there are constants $c_3$, $c_4$ and $c_5$ (depending on $\eta$), such that for any positive integer $q$ and any $\lambda>1$
\begin{enumerate}
\item[{\rm (i)}] if $q>c_3\lambda^{3+\eta}$ then the Glauber dynamics of the $q$-state Potts model
at activity $\lambda$ mixes rapidly for the class of connected graphs of maximum degree 4;
\item[{\rm (ii)}] if $q>c_4\lambda^{2+\eta}$ then the Glauber dynamics of the $q$-state Potts model
at activity $\lambda$ mixes rapidly for the toroidal grid;
\item[{\rm (iii)}] if 
$q<c_5\lambda^{\tfrac{8}{3}-\eta}$ then the Glauber dynamics of the $q$-state
Potts model at activity $\lambda$ mixes slowly for almost all regular graphs of degree $4$.
\end{enumerate}
In particular, for sufficiently large $\lambda$ there is a positive integer 
$q$ such that the Glauber dynamics of
the $q$-state Potts model at activity $\lambda$ mixes rapidly for the toroidal grid, but slowly
for almost all regular graphs of degree $4$.
\end{theorem}

The purpose of Theorem~\ref{th:main3} is illustrative and it is proved at the end of the paper. 
Theorem~\ref{th:main3}(i) and (iii) are immediate consequence of Theorem~\ref{th:main2} (by substituting $\Delta=4$), while
Theorem~\ref{th:main3}(ii) is a useful illustration of our general technique applied to the grid.
 A more detailed statement of 
Theorem~\ref{th:main3}(ii) is given 
as part of Theorem~\ref{th:glaubergrid}.

Section~\ref{se:rapid} contains our results on rapid mixing of Glauber dynamics.
Section~\ref{se:extremal} is devoted to an 
extremal problem whose solution allows us to obtain improved bounds for our slow-mixing results in Section~\ref{se:tormix}.

\subsection{Comparison with related results and phase transitions}

We write $o(1)$ for an expression 
that tends to $0$ as $q \to \infty$.
(The most interesting setting for our results is when $q$ is large.)
We now restate our results in terms of the inverse temperature $\beta$, under the assumption that
$J=1$, so that $\lambda = e^\beta$. 

The results of Theorems~\ref{th:main1}, \ref{th:main2}(i), \ref{th:main2}(ii), \ref{th:main3}(ii) say respectively:
\begin{itemize}
\item[(a)]  if $\beta \leq \frac{1 + o(1))}{\Delta}\, \log q$ then the Glauber dynamics of the $q$-state Potts model mixes rapidly on  graphs of maximum degree $\Delta$;
\item[(b)] if $\beta \leq \frac{1 + o(1)}{\Delta-1}\, \log q$ then the Glauber dynamics of the $q$-state Potts mixes rapidly on graphs of maximum degree $\Delta$;
\item[(c)] if $\beta > \frac{1 + o(1)}{\Delta - 1 - \nfrac{1}{\Delta-1}}\, \log q $ then the Glauber dynamics of the $q$-state Potts model  mixes slowly for almost all regular graphs of
degree $\Delta \geq 3$;
\item[(d)] if $\beta < \frac{1+o(1)}{2}\, \log q$ then the Glauber dynamics of the $q$-state Potts model mixes rapidly for toroidal grids.
\end{itemize}

There is some overlap between Theorem~\ref{th:main1}
and a result of Hayes~\cite[Proposition 14]{hayes} for $q=2$,
which was generalised to arbitrary $q$ by 
Ullrich~\cite[Corollary 2.14]{ull}.
Ullrich showed that when the inverse temperature $\beta$
satisfies $\beta \leq 2c/\Delta$ for some $0<c<1$, then the
Glauber dynamics 
is rapidly mixing on graphs of maximum degree $\Delta$.
Hence our result (a) holds for a wider range of $\beta$ when
$q$ is large. (For small values of $q$, 
Theorem~\ref{th:main1} does not apply but~\cite[Corollary 2.14]{ull} is valid).

As we have mentioned, there is often a link between certain phase transitions and the critical inverse temperature of associated dynamics (i.e.\ an inverse temperature below which the dynamics mix rapidly and above which they mix slowly). We will not define what we mean by phase transitions here but mention only that, for Glauber dynamics of the $q$-state Potts model on a random $\Delta$-regular graph, the relevant phase transition is the transition from unique to non-unique Gibbs measure on the infinite $\Delta$-regular tree. H{\"a}ggstr{\"o}m \cite{Hag} showed that this phase transition occurs at an inverse temperature $\beta_0 = \log B$, where $B$ is the unique value for which the polynomial 
\[
(q-1)x^{\Delta} + (2 - B - q)x^{\Delta - 1} +Bx - 1
\]
has a double root in $(0,1)$. While there is no general closed form formula for $\beta_0$, we show in the appendix that $\beta_0 = \frac{\log q}{\Delta - 1} + O(1)$.
 Thus $\beta_0$ approximately matches the rapid mixing bound of (b). 
 
We note that, in a recent related work, Galanis \textit{et al.}~\cite{GSVY} give a very detailed picture of the phase transitions of the ferromagnetic Potts model on the infinite $\Delta$-regular tree. Using this analysis they also show that show that the Swendsen-Wang process (a MCMC process different to Glauber dynamics) mixes slowly at a specific phase transition point on almost all random regular graphs of degree $\Delta$.

As mentioned earlier, result (d) is only illustrative  
since sharper bounds for the grid are known.  It is known that for the 
infinite 2-dimensional grid, the phase transition occurs at $q=(\lambda-1)^2$~\cite{wu82}
(i.e.\ $\beta = \log(1 + q^\frac{1}{2})$) 
and that rapid mixing occurs for finite grids when 
$\beta$ is below this threshold; see~\cite{mar99} 
and Theorem~2.10 of \cite{ull}. 
It is conjectured that the Glauber dynamics mixes slowly when
$\beta$ is above this threshold;
see Remark~2.11 of \cite{ull}). 
Borgs, Chayes and Tetali~\cite{BCT} proved that for $q$ sufficiently large and for 
$\beta > \frac{\log\left(q\right)}{2}+O(q^{-1/2})$, the heat bath Glauber dynamics is slowly mixing on sufficiently large
toroidal grids (with a mixing time exponential in $\beta$ and in $L$, the side length of the grid).
This improved on the earlier result~\cite{bor99}.

\section{Mixing time upper bounds}
\label{se:rapid}

Our goal in this section is to give good lower bounds on the number of colours needed for the Glauber dynamics to mix rapidly. We begin by describing the notions of coupling and path coupling, which are very useful tools in proving upper bounds on mixing times for Markov chains. In Section~\ref{se:GD}, we apply path coupling directly to the Glauber dynamics of bounded-degree graphs to obtain our first lower bound on the number of colours needed for rapid mixing. In Section~\ref{se:BD}, we consider block dynamics, a more general type of dynamics that can be used to sample from the Gibbs distribution. We give a general lower bound on the number of colours needed for rapid mixing of block dynamics (Theorem~\ref{le:block}). 
We illustrate how to apply Theorem~\ref{le:block} to 
bounded-degree graphs 
in Section~\ref{se:grid}. In Section~\ref{se:comparison}, we relate the mixing times of Glauber dynamics to that of the block dynamics and show how this gives various improvements to the bounds obtained in Section~\ref{se:GD}.
This enables us, in Theorems~\ref{th:glaubercompare} and~\ref{th:glaubergrid}, to prove what is needed for Theorem~\ref{th:main2} part (i), and Theorem~\ref{th:main3} parts (i) and (ii). Note that the final proofs of Theorems~\ref{th:main2} and~\ref{th:main3} are left until we have all the pieces, at the end of Section~\ref{se:tormix}.

\subsection{Coupling}
\label{se:coupling}

The notion of coupling (more specifically path coupling \cite{BubleyDyer}) lies at the heart of our proofs of upper bounds
for mixing times.
We give the basic setup in this section.

Let $\mathcal{M} = (X_t)$ be a Markov chain with transition matrix $P$. A \emph{coupling} for $\mathcal{M}$ is a stochastic process  $(A_t,B_t)$ on $\Omega \times \Omega$ such that each of $(A_t)$ and $(B_t)$, considered independently, is a faithful copy of $(X_t)$. Since all our processes are time-homogeneous,  a coupling is determined by its transition matrix: given elements $(a,b)$ and $(a',b')$ of $\Omega \times \Omega$, let $P'((a,b),(a',b'))$ be the probability that $(A_{t+1},B_{t+1}) = (a',b')$ given that $(A_t,B_t) = (a,b)$. Since $(A_t,B_t)$ is a coupling, for each fixed $(a,b) \in \Omega \times \Omega$, we have 
\begin{align*}
\sum_{b' \in \Omega}P'((a,b),(a',b')) &= P(a,a') \quad \text{ for all } a' \in \Omega; \\
\sum_{a' \in \Omega}P'((a,b),(a',b')) &= P(b,b') \quad \text{ for all } b' \in \Omega  .
\end{align*}
Under \emph{path coupling}, the coupling is only defined on a subset $\Lambda $ of $ \Omega \times \Omega$. This restricted coupling is then extended to a coupling on the whole of $\Omega \times \Omega$ along paths in the state
space $\Omega$. 
In our setting, we have $\Omega = [q]^V$, where $V$ is the vertex set of some fixed graph. For $\sigma, \sigma' \in \Omega$, we write $d(\sigma, \sigma')$ for the number of vertices on which $\sigma$ and $\sigma'$ differ in colour (that is, the Hamming distance).  Define $\Lambda \subseteq \Omega \times \Omega$ by 
\[ \Lambda = \{(\sigma,\sigma'): d(\sigma,\sigma')=1 \}.\]
The key property of $\Lambda$ required for the path coupling
method is that for any $\sigma,\sigma'\in\Omega$,
by recolouring the $d(\sigma,\sigma')$ disagreeing
vertices one by one in an arbitrary order, we obtain a path of
length $d(\sigma,\sigma')$ from $\sigma$ to $\sigma'$, with consecutive
elements of the path corresponding to an element of $\Lambda$.
  
\begin{lemma} [See~\cite{DG98} for example]
\label{le:coup}
Let $\Omega=[q]^V$ and $\Lambda$ be as above, 
with $n:=|V|$,
and let $\mathcal{M}$ be some Markov chain on $\Omega$. 
Suppose that we can define a coupling 
$(A,B)\mapsto (A',B')$ for $\mathcal{M}$ on $\Lambda$ such that for some constant $\beta < 1$ and all $(A,B)\in \Lambda$ we have
\[ \E(d(A',B')\mid (A,B)) \leq \beta .\]
  Then by path coupling we may conclude that 
\[ \tau(\mathcal{M},\varepsilon) \leq \frac{\log(n\, \varepsilon^{-1})}{1-\beta}.\]
\end{lemma}

\subsection{Glauber dynamics} 
\label{se:GD}

Our goal in this subsection is to prove Theorem~\ref{th:main1}. 
In the subsections that follow, we shall see how we can improve Proposition~\ref{le:vertex} in some special cases, but in Section~\ref{se:tormix}, we shall see that the bound given below is close to best possible, at least
in terms of the exponent of $\lambda$.

We actually prove the following proposition, which immediately implies
Theorem~\ref{th:main1} but also provides a
bound on the mixing time. The proof is a standard coupling calculation.

\begin{proposition} \label{le:vertex}
Let $G$ be a graph with maximum degree $\Delta$, and fix an activity $\lambda>1$. Suppose that $q$ is an integer which satisfies $ q \geq \Delta\lambda^\Delta + 1$. 
Recall that $\mathcal{M}_{\mathrm{GD}} = 
\mathcal{M}_{\mathrm{GD}}(G, \lambda, q)$ denotes the Glauber dynamics
for the $q$-state Potts model on $G$ at activity $\lambda$.
Then
\[ \tau(\mathcal{M}_{\mathrm{GD}},\varepsilon) \leq 
    (\Delta+1)\, n\log(n\, \varepsilon^{-1}).
\]
\end{proposition}

\begin{proof}
Fix $(A,B)\in\Lambda$ and let $u$ be the (unique) vertex which is coloured differently by $A$ and $B$. 
We define a coupling $(A,B)\mapsto (A',B')$ as follows:  let $\vec{v}$
be a uniformly random vertex of $G$, and given that $\vec{v}=v$, obtain $A'$ (respectively, $B'$) by updating the colour of the vertex $v$ in $A$ (respectively, $B$) according to the distributions $\phi_A := \phi_A^v$ (respectively, $\phi_B := \phi_B^v$). The joint distribution on $(\phi_A, \phi_B)$ is chosen so as to maximise the probability that $A'(v) = B'(v)$. Call this maximised probability $p=p(v,A,B)$. It is not hard to see that
\[ 
1-p = \frac{1}{2}\sum_{c \in [q]} |\phi_A(c) - \phi_B(c)| = \|\phi_A - \phi_B \|_{TV}.
\]
Observe that $p(v,A,B) = 1$ if $v=u$ or if $v$ is not a neighbour of $u$ (because in both cases, $A$ and $B$ assign the same colours to the neighbours of $v$ and so $\phi_A$ and $\phi_B$ are the same distribution). 

Now assume that $v$ is a neighbour of $u$, so that  $\phi_A$ and $\phi_B$ are different distributions. Without loss of generality, we may assume that $A(u)=1$ and $B(u)=2$. 
Let $a_i:=n(A,v,i)$, that is,  $a_i$ is the number of neighbours of $v$ 
coloured $i$ by $A$.  Similarly, let $b_i := n(B,v,i)$. Note that $b_1 = a_1 - 1$, $b_2 = a_2+1$ and $b_i = a_i$ for  $i = 3, \ldots, q$. Define 
\[
Z_A = \sum_{i=1}^q \lambda^{a_i} \:\:\:\:\: \text{and} \:\:\:\:\:
Z_B = \sum_{i=1}^q \lambda^{b_i} \,=\, Z_A + (1- \lambda^{-1})(\lambda^{a_2+1} - \lambda^{a_1}),
\]
and assume without loss of generality that $Z_B \leq Z_A$. 
It is easy to see that $\phi_A(i) \leq \phi_B(i)$ for $i=2,\ldots, q$ and hence $\phi_A(1) \geq \phi_B(1)$. Thus
\[ \|\phi_A - \phi_B\|_{TV} = \max_{R \subseteq [q]}|\phi_A(R) - \phi_B(R)| = |\phi_A(1) - \phi_B(1)| = 
\frac{\lambda^{a_1}}{Z_A} - \frac{\lambda^{b_1}}{Z_B}. 
\] 
Given $\vec{a} = (a_1, \ldots, a_q) \in [\Delta]^q$, define $f(\vec{a},\lambda,q) = \frac{\lambda^{a_1}}{Z_A} - \frac{\lambda^{b_1}}{Z_B}$, and let $g(\lambda,q)$ be the maximum of $f$ over all $\vec{a} \in [\Delta]^q$ subject to $a_1 + \cdots + a_q = \Delta$.

Observe that
\begin{align*} 
\E\left(d(A',B')-1\big| (A,B)\right) 
&= (-1)\Prob(\vec{v}=u) + \sum_{v \in N(u)}\Prob(\vec{v}=v)(1-p(v,A,B)) \\
&\leq -\frac{1}{n} + \frac{\Delta}{n} g(\lambda,q).
\end{align*}
We give an easy upper bound for $g(\lambda, q)$ as follows. First, for all $\vec{a}\in [\Delta]^q$ we have
\[ 
	f(\vec{a}, \lambda, q) \leq \frac{\lambda^{a_1}}{Z_A}.
\]
The right hand side of the above is increasing in all directions of the 
form $e_1 - e_i$, where $e_1, \ldots, e_q$ is the standard basis for 
$\field{R}^q$. Therefore the right hand side is maximised when 
$\vec{a} = (\Delta,0,\ldots, 0)$ giving 
\[ g(\lambda,q) \leq \frac{\lambda^\Delta}{\lambda^\Delta + q-1} \leq 
    \frac{1}{\Delta+1},
\]
using the lower bound on $q$ to obtain the final inequality.
Therefore. 
\begin{align*} 
  \E\left(d(A',B')\big| (A,B)\right) \leq 
   1+ \frac{1}{n}\left(-1 + \frac{\Delta}{\Delta+ 1}\right)
   &= 1- \frac{1}{(\Delta+1)n}.
\end{align*}
Applying Lemma~\ref{le:coup} completes the proof.
\end{proof}

\subsection{Block dynamics}
\label{se:BD}

In this section we begin the analysis of \emph{block dynamics} in which, at each step, the colours of several vertices (or a \emph{block} of vertices) are updated. We first present the framework and show general results on block dynamics. In the next subsection we discuss suitable choices of blocks and, in 
Theorem~\ref{pr:3items}, show rapid mixing of block dynamics for certain block systems. 

As before, let $G=(V,E)$ be a graph, fix $\lambda >1$ and let $\Omega = [q]^V$, 
where $[q] = \{1, \ldots, q\}$.
Let $\mathscr{S} = \{ S_1,\ldots, S_R\}$ be a collection of subsets of $V$
such that $\cup_{S\in\mathscr{S}} S = V$.  Each element of $\mathscr{S}$ 
is called a \emph{block}, and we call $\mathscr{S}$ a \emph{block system} 
for $G$.  
Fix a probability distribution $\psi$ on $\mathscr{S}$.
We define a Markov chain $\mathcal{M}_{\mathrm{BD}}=\mathcal{M}_{\mathrm{BD}}^{\mathscr{S}, \psi}(G, \lambda, q)$ with state space $\Omega$, which we call the $(\mathscr{S}, \psi)$-\emph{block dynamics}. We ensure that the new chain also has the Gibbs distribution as its stationary distribution.
First we need some more notation.

Given $S\in \mathscr{S}$, for $c\in [q]^{S}$ and $X\in\Omega$ we
let $X^{(S,c)}\in \Omega$ be the
configuration defined by
\[ X^{(S,c)}(u) = \begin{cases} X(u) & \text{ if $u\not\in S$,}\\
         c(u) & \text{ if $u\in S$.}\end{cases}
\]
Let $\mu_{X,S}(c)$ denote the number of monochromatic edges in $X^{(S,c)}$
which are incident with at least one vertex of $S$.  
Finally, define the distribution $\phi_{X,S}$ on $[q]^{S}$ by
\[ \phi_{X,S}(c) \propto \lambda^{\mu_{X,S}(c)}, 
\quad \text{ that is, } \,\,\,  
\phi_{X,S}(c) = \frac{\lambda^{\mu_{X,S}(c)}}{Z_{X,S}}
\]
where 
\[
Z_{X,S} = \sum_{c \in [q]^{S}}  \lambda^{\mu_{X,S}(c)}.
\]
The transition procedure of the $(\mathscr{S},\psi)$-block dynamics can now be 
described.
From current state $X_t\in\Omega$, obtain the new state $X_{t+1}\in \Omega$ as follows:
\begin{itemize}
\item choose a random $\vec{S} \in \mathscr{S}$ according to the distribution $\psi$;
\item given that $\vec{S}=S$, choose a configuration $c \in [q]^{S}$ for $S$ from the distribution 
$\phi_{X_t,S}$;
\item let $X_{t+1} = {X_t}^{(S,c)}$.
\end{itemize}
The stationary distribution of this chain is the Gibbs distribution on $\Omega$.

Theorem~\ref{le:block} below gives a sufficient condition on the 
number of colours
for the $(\mathscr{S},\psi)$-block dynamics to be rapidly mixing.
The result is stated in terms of three parameters which we now define.

For $S \subseteq V$, write $\partial S$ for the set of vertices in $V \setminus S$ that have a neighbour in $S$. 
Write $s:= \max_{S \in \mathscr{S}}|S|$ for the size of the largest
block in $\mathscr{S}$.  
Let $\vec{S} \in\mathscr{S}$ be a random block chosen according to the 
distribution $\psi$. Given $v\in V$, define
\[ \psi(v) = \Prob(v \in \vec{S}), \quad 
   \psi_{\partial}(v) = \Prob(v \in \partial \vec{S}).
\]
Our first parameter $\partial^+$ is
\begin{equation}
\label{partial-plus}
 \partial^+ = \partial^+(\mathscr{S}) =
  \max_{S\in\mathscr{S}} |\partial S|^{\min\{ |S|,\, |\partial S|\}}.
\end{equation}
Let $\psi_{\min} := \min_{v \in V} \psi(v)$ and 
define our second parameter $\Psi$ by
\begin{equation}
\label{Psi}
 \Psi = \Psi(\mathscr{S}, \psi) = \max_{v \in V} \frac{\psi_\partial(v)}
                                         {\psi(v)}.
\end{equation}
These first two parameters are in some sense less important than the third parameter since they are essentially used as crude estimates for quantities that we do not aim to control too precisely.

For the third parameter we require some terminology.
Given $A \subseteq V$ and $X\in \Omega$,  write $X|_A$ for the
configuration $X$ restricted to $A$.
Consider a configuration
$c\in [q]^S$.  
A colour used by $c$ is called \emph{free} with respect to $X,S$  if it does not appear in $X|_{\partial S}$.   
Write $f(X,S,c)$
 for the number of free colours in $c$ with respect to $X,S$.
For our third parameter, we first define for each positive integer $f$ 
\[
 \mu^+_{X,S,f}  =
   \max \left\{
    \frac{\mu_{X,S}(c)}{|S| - f} : c\in [q]^S,\,\, f(X,S,c)=f\right\},
\]
where the maximum over an empty set 
is defined to be zero.
We set 
\begin{equation}
\label{mu-plus}
 \mu^+ = \mu^+(\mathscr{S}) = \max_{S\in\mathscr{S}}\, \max_{X\in \Omega} \,  
    \max_{f=0,\ldots, |S|-1}\, 
 \mu^+_{X,S,f}.
\end{equation}
Although the definition of $\mu^+$ gives an a priori dependency on $q$, in all our 
applications on bounded-degree graphs we can bound $\mu^+$ independently 
of $q$  (see Proposition~\ref{pr:boundd}).
Hence we suppress this dependence in our notation. 

Let us sketch a very informal argument to show that block dynamics mixes rapidly roughly when $q \geq \lambda^{\mu^+}$; this will be formalised in the statement and proof of Theorem~\ref{le:block}.
Fix $X \in \Omega$ and $S \in \mathscr{S}$, where $|S|$ is typically thought of as a small number and $q$ a large number. We are interested in estimating the quantity $q^{|S|} / Z_{X,S}$, which, in the distribution $\phi_{X,S}$, is approximately the probability of choosing a \emph{free} configuration for $S$.  
A free configuration is one in which each vertex in $S$ receives a distinct free colour, so that $S$ is coloured with $|S|$ free colours in total. If this probability is close to $1$ for all choices of $X,S$ then, intuitively at least, one expects the block dynamics to mix rapidly.

To show $q^{|S|} / Z_{X,S}$ is close to $1$, we must show that the contribution of non-free configurations to $Z_{X,S}$ is relatively small (compared to $q^{|S|}$). Consider the contribution from configurations with a fixed number $f \leq |S|-1$ of free colours. There are approximately $q^f$ such configurations $c$, each contributing 
$\lambda^{\mu_{X,S}(c)} \leq \lambda^{(|S|-f)\mu^+_{X,S,f}}$ to $Z_{X,S}$, giving a total contribution of at most $q^f \lambda^{(|S|-f)\mu^+_{X,S,f}}$. Comparing to $q^{|S|}$ gives
\[
q^{|S|}/q^f \lambda^{(|S|-f)\mu^+_{X,S,f}} = [q \lambda^{-\mu^+_{X,S,f}}]^{|S|-f} \geq q \lambda^{-\mu^+_{X,S}}.
\]
This last expression is at least $1$ provided $q > \lambda^{\mu^+_{X,S,f}}$, and this inequality holds for all choices of $X,S,f$ if $q > \lambda^{\mu^+}$. From these crude calculations we expect rapid mixing of block dynamics roughly when $q > \lambda^{\mu^+}$.

The following theorem formalises the argument above, giving a sufficient condition on the number 
of colours for  $(\mathscr{S},\psi)$-block dynamics to be rapidly mixing.

\begin{theorem} \label{le:block}
Let $G=(V,E)$ be a connected graph and let 
$\mathscr{S}$ be a block system for $G$
such that $V\not\in\mathscr{S}$. 
Let $\psi$ be a distribution on $\mathscr{S}$ and fix $\lambda > 1$. If 
\[ q \geq (2s)^{s+1}\,  \partial^+\, \Psi\, \lambda^{\mu^+}\]
\emph{(}where parameters 
$s$,
$\partial^+$, $\Psi$ and $\mu^+$ are as defined above\emph{)}
then the $(\mathscr{S},\psi)$-block dynamics
$\mathcal{M}_{\mathrm{BD}} = 
  \mathcal{M}_{\mathrm{BD}}^{\mathscr{S},\psi}(G, \lambda, q)$
satisfies 
\[ \tau(\mathcal{M}_{\mathrm{BD}},\varepsilon) \leq 2\psi_{\min}^{-1}\log(n \varepsilon^{-1}). \]
\end{theorem}

We remark that for the bound $q \geq (2s)^{s+1}\,  \partial^+\, \Psi\, \lambda^{\mu^+}$ in Theorem~\ref{le:block}, we expect the constant multiplicative factor $(2s)^{s+1}  \partial^+ \Psi$ can be improved; however we have not attempted to do this in order to keep our treatment simple.

\begin{proof}
We define a coupling $(A,B)\mapsto (A',B')$ for $\mathcal{M}_{\mathrm{BD}}$ on $\Lambda$ as follows. Given $(A,B)\in\Lambda$, let $u=u(A,B)$ be the (unique) vertex which is coloured differently by $A$ and $B$. We choose a random $\vec{S} \in \mathscr{S}$ using the distribution $\psi$, 
and given that $\vec{S}=S$,
we obtain $A'$ (respectively, $B'$) by updating the colouring of $S$ in $A$ (respectively, $B$) according to the distribution $\phi_A := \phi_{A,S}$ (respectively, 
$\phi_B := \phi_{B,S}$); this will give a coupling since $A$ and $B$ are updated using the transition procedure of $\mathcal{M}_{\mathrm{BD}}$. 
 We choose the joint distribution on $(\phi_A, \phi_B)$ so as to maximise the probability that 
$A'|_S = B'|_S$. Call this maximised probability $p(S,A,B)$. 
Observe that $p(S,A,B) = 1$ if $u \not\in \partial S$ (because $A$ and $B$ assign the same colours to $\partial S$,
 so $\phi_A$ and $\phi_B$ are the same distribution). For the case that $u \in \partial S$, we uniformly bound $p(S,A,B)$ by setting 
\[ p := \min_{(A,B)\in\Lambda}\, \min_{S\in \mathscr{S} : u \in \partial S} 
  \, p(S,A,B).
\]
(Let $p=1$ if, for all $S\in\mathscr{S}$, $u\not\in\partial S$.)
Now for all $S\in\mathscr{S}$ with $u\in\partial S$ we have
\begin{align}\label{eq:1}
 p(S,A,B) = \sum_{c \in [q]^S} \min(\phi_A(c), \phi_B(c)) &\geq  \sum_{c \in [q]^S} \frac{1}{\max(Z_{A,S}, Z_{B,S})} \nonumber \\
&= \frac{q^{|S|}}{\max(Z_{A,S}, Z_{B,S})}.
\end{align}
We claim that
\begin{equation}
\label{fred}
\frac{q^{|S|}}{Z_{X,S}} \geq 1 - \frac{1}{2s\Psi}
\end{equation}
for all $X\in\Omega$ and $S\in\mathscr{S}$.  If  (\ref{fred})
holds then substituting into (\ref{eq:1}) gives
\[ p\geq 1 - \frac{1}{2s\Psi},\]
which in turn implies that
\begin{align*}
\E\left(d(A',B') - 1\big| (A,B)\right) 
&= -\Prob(u \in \vec{S}) + \sum_{S \in \mathscr{S}:\, u \in \partial S}\Prob(\vec{S}=S)|S|(1-p(S,A,B)) \\
&\leq 
     \, {} - \psi(u) + s\, \psi_{\partial}(u)\, (1-p) \\
 &= -\psi(u) \left( 1 - \frac{s\psi_{\partial}(u)}{\psi(u)}(1-p) \right)\\
&\leq -\psi_{\min} \left( 1 - s \Psi\, (1-p) \right) \\
&\leq -\frac{\psi_{\min}}{2}.
\end{align*}
The theorem follows from this, by Lemma~\ref{le:coup}.
So it remains to establish (\ref{fred}).

Fix $X\in\Omega$ and $S\in\mathscr{S}$.  
For any configuration $c$, write $Q(c)$ for the set of colours used by $c$.  
Given a configuration
$c\in [q]^S$, the colour classes of $c$ define a partition $P$ of
$S$ into (unordered) nonempty parts. 
(Here, we think of a partition $P$ of $S$ as a set of nonempty parts $\{ P_1, \ldots, P_t\}$ where $P_i \subseteq S$ are disjoint and $\cup_{A \in P}A=S$.) 
 Let $F\subseteq P$ be the
set of colour classes corresponding to colours which are free with
respect to $X,S$ (in the given configuration $c$). 

Conversely, we can start from a partition $P$ of $S$ 
and a subset $F$ of $P$.  Given a set of $|P|$ colours, we can
form a configuration of $S$ by assigning a distinct colour to each part of $P$
such that the colour assigned to $A\in P$ belongs to 
$[q]\setminus Q(X|_{\partial S})$ if and only if $A\in F$.  Any configuration
which
can be formed in this way is called a $(P,F)$-\emph{configuration} of $S$.
(Such a configuration is uniquely determined by $(P,F)$ and the map $P\to [q]$
which performs the assignment of colours.)

Let $n(S,P,F)$ be the number of $(P,F)$-configurations of $S$.
By definition of $\mu^+$ we have
\[ Z_{X,S} = \sum_{c\in [q]^S} \lambda^{\mu_{X,S}(c)}
      \leq q^{|S|} + 
     \sum_{(P,F) : |F|\neq |S|} \, n(S,P,F)\, \lambda^{(|S|-|F|)\mu^+}.
\]
The first term corresponds to $P=F$ with $|P|=|S|$, arising from
a configuration $c\in [q]^S$ in which every vertex in $S$ receives
a distinct free colour. 
(These were called ``free configurations'' in the sketch proof.)
We use $q^{|S|}$ as an upper bound for the
number of such configurations.  For all other values of $(S,P,F)$ we
have the following crude bound:
\[ n(S,P,F) \leq q_1^{\min\{q_1, |P|-|F|\}}\, (q-q_1)^{|F|} 
      \leq |\partial S|^{\min\{ |S|, |\partial S|\}} \, q^{|F|}
       \leq \partial^+\, q^{|F|}, 
\]
where $q_1 = |Q(X|_{\partial S})|$ and we recall that all parts must be coloured differently.
Substituting gives
\[
Z_{X,S} \leq q^{|S|} + \sum_{(P,F):\, |F|\neq |S|} \,
         \partial^+\, q^{|F|}\,  \lambda^{(|S|-|F|)\mu^+} .
\]
Now applying the bound on $q$ from the theorem statement gives
\begin{align}
\frac{Z_{X,S}}{q^{|S|}} &\leq 1 + \sum_{(P,F) : |F|\neq |S|} \,
               \partial^+\, q^{|F|-|S|} \lambda^{(|S|-|F|)\mu^+}\notag\\
  &\leq 1 + \sum_{(P,F) : |F|\neq |S|} \partial^+\,
                          ((2s)^{s+1}\,  \partial^+\, \Psi\, 
       \, \lambda^{\mu^+})^{|F|-|S|}\, \lambda^{(|S|-|F|)\mu^+}\notag\\
 & \leq  1 +  \sum_{(P,F):\, |F| \not= |S|} \,
        ((2s)^{s+1}\, \Psi )^{|F|-|S|}.
  \label{interrupted}
\end{align}
The number of terms in the above sum is at most $(2|S|)^{|S|}$, since there
are at most $|S|^{|S|}$ choices of the partition $P$ and at
most $2^{|P|} \leq 2^{|S|}$ choices of $F$.  

Next, note that
\[ 
 \Psi  = \max_{v \in V} \frac{\psi_\partial(v)}
                                         {\psi(v)}
 \geq \E_\rho\left(\frac{\psi_\partial(v)}{\psi(v)}\right)
        =  \sum_{v\in V} \rho(v)\, \frac{\psi_\partial(v)}{\psi(v)}
\]
for any probability distribution $\rho$ on $V$.  In particular,
we can take $\rho(v) = \psi(v)/N$, where
\[ N = \sum_{v\in V}\psi(v) =  \sum_{S\in\mathscr{S}}\, \psi(S)\, |S|
               \leq s.
\]
With this choice of $\rho$, we obtain the bound
\begin{align*}
\Psi &\geq N^{-1}\, \sum_{v\in V} \psi_\partial(v)
         = N^{-1}\, \sum_{S\in S}\, \psi(S)\, |\partial S|
     \geq s^{-1}
\end{align*}
since $\partial S$ is nonempty for all $S\in\mathscr{S}$, 
as $G$ is connected and $V\not\in\mathscr{S}$.  
It follows that $(2s)^{s+1}\, \Psi > 1$,  
and combining this with (\ref{interrupted}) gives
\begin{align*}
\frac{Z_{X,S}}{q^{|S|}} 
  & \leq 1 + \frac{1}{2s\Psi}.
\end{align*}
Inverting this and using the identity $(1+y)^{-1}\geq 1-y$
establishes (\ref{fred}), completing the proof.
\end{proof}

\subsection{Block dynamics for specific examples}
\label{se:grid}

In this subsection we illustrate how one can use Theorem~\ref{le:block} to obtain rapid mixing results for block dynamics on graphs of bounded degree. 
In the next subsection, we shall see how these results for block dynamics can be translated into rapid mixing results for Glauber dynamics.  

In order to build some intuition,
 we begin by investigating the range of possible values of the parameter
$\mu^+$. 
We will need the following notation: given $T\subseteq T' \subset V$, we write $\vol(T,T')$ for the set of edges of $G$ that are contained in $T'$ and have at least one endvertex in $T$.

\begin{proposition}
\label{pr:boundd}
Let $G=(V,E)$ be a graph of maximum degree $\Delta$ and let
$\mathscr{S}$ be any block system for $G$. Then 
\[ \mu^+ = \mu^+(\mathscr{S}) \leq \Delta.\]
If in addition $G$ is regular then
\[ \frac{\Delta}{2}\leq \mu^+(\mathscr{S}) \leq \Delta.\]
\end{proposition} 

\begin{proof}
First fix $X\in\Omega$ and $S\in \mathscr{S}$.
Given a configuration $c\in [q]^{S}$, let $P$ be 
the partition of $S$ defined by the nonempty colour classes of $c$.
Define $F \subseteq P$ to be the set of colour classes of $c$ which
correspond to a colour which does not appear on $X|_{\partial S}$.
Let 
\[ 
   A_F = \bigcup_{A \in F} \, A
\]  
and 
\[ 
   A'_F = \bigcup_{A \in F\, :\, |A|\geq 2} \, A.
\]  
Since $G$ has maximum degree $\Delta$, a trivial upper bound on $\mu_{X,S}(c)$ is $\Delta|S|$. 
But note that if a monochromatic edge $e$ is incident to a 
 vertex in $A_F$, then $e$ must have both endpoints in the same part $A$ of $F$. Thus edges incident to vertices in $A_F \setminus A'_F$ do not contribute to $\mu_{X,S}(c)$
 and monochromatic edges incident to vertices in $A'_F$ are double counted in the trivial bound. Hence
\begin{align*}
\mu_{X,S}(c)
         & \leq \Delta(|S|-|A_F|)+\frac{\Delta}{2}\,|A'_F| \\
         & = \Delta |S| - \Delta\left(|A_F|-\frac{|A'_F|}{2}\right) \\
         &\leq  \Delta (|S| - |F|).
\end{align*}
Hence the upper bound holds, by definition of $\mu^+$.

Next, suppose that $G$ is $\Delta$-regular with $X \in \Omega$ and $S\in\mathscr{S}$.
Consider any configuration $c\in [q]^S$ which assigns a single
colour to all of $S$, 
and where this is the only colour used in $X|_{\partial S}$.
Then
\[ \mu^+ \geq \frac{\mu_{X,S}(c)}{|S|-|F|} = 
\frac{|\vol(S, S \cup \partial S)|}{|S|} \geq \frac{\Delta}{2},
\]
where the last inequality follows because $G$ is regular of degree $\Delta$.

\end{proof}

Next we show how to improve the upper bound on $\mu^+$ given in Proposition~\ref{pr:boundd}  by choosing our block system more carefully. 

Let $k\geq 2$ be an integer and let $G=(V,E)$ be a graph with
$n$ vertices and with maximum degree $\Delta$.
Let 
\[ \mathscr{S} = \{ S_v : v\in V\} \]
where for all $v\in V$ the set $S_v\subseteq V$ satisfies
$ v\in S_v,\,\,\, |S_v|=k$ and $G[S_v]$ is connected. Then $\mathscr{S}$ is called a $k$-\emph{block system} for $G$.
Let $\psi$ be the uniform distribution over $\mathscr{S}$.
To apply Theorem~\ref{le:block} to the $(\mathscr{S},\psi)$-block
dynamics we will calculate upper bounds on 
the parameters $\partial^+$, $\Psi$ and $\mu^+$.

Clearly $|\partial S|\leq \Delta k$ and $\min\{ k,|\partial S|\}\leq k$
for all $S\in\mathscr{S}$.  Hence
\begin{equation}
\label{partial-plus-k}
\partial^+ \leq (\Delta k)^k.
\end{equation}
To compute $\Psi$, observe first that $\psi(v) \geq 1/n$ for all $v \in V$ as there are $n$ blocks and each vertex belongs to at least one block. 
Next, observe that $\psi_\partial(v) \leq \frac{\Delta^k}{n}$:
indeed if $v\in\partial S_u$ for some $u \in V$ then $u$
is at distance at most $k$ from $v$ and  since
and there are at most $\Delta^k$ vertices (excluding $v$) at distance at most $k$ from $v$ in $G$, there are at most $\Delta^k$ out of $n$ blocks containing $u$ in their boundary.
Therefore
\begin{equation}
\label{psi-k}
 \Psi = \max_{v\in V} \, \frac{\psi_\partial(v)}{\psi(v)} \leq \Delta^k.
\end{equation}
In order to calculate an upper bound on $\mu^+$ we first prove a
preliminary result.  For $T\subseteq T'\subset V$, recall the notation 
$\vol(T,T')$ introduced above
Proposition~\ref{pr:boundd}, and note that $\vol(T,T)$ is just the set of 
edges inside $T$.  

For any two sets $A,B$, we write $\delta_{A,B}$ for the indicator function that $A=B$, that is $\delta_{A,B} =1$ if $A=B$ and $\delta_{A,B} =0$  otherwise.  

\begin{proposition}
\label{pr:vol}
Let $H=(V,E)$ be a connected graph and let $U \subseteq V$. Then 
\[ |\vol(U,V)| \geq |U| - \delta_{U,V}.\]
\end{proposition}

\begin{proof}
It is sufficient to prove the statement for $H$ a tree. The statement is clear if $U=V$.  Now suppose that $U\neq V$ and consider the components $C_1, \ldots, C_r$ of $H[U]$. Then $\vol(C_i, V)$ has at least $|C_i|$ edges and is disjoint from $\vol(C_j,V)$ for all $j \not= i$. Thus 
\[
 |\vol(U,V)| = \sum_{i=1}^r |\vol(C_i,V)| \geq \sum_{i=1}^r|C_i| = |U|.
\]
\end{proof}

Next we give an upper bound on the parameter $\mu^+$ for $k$-block
systems.  For $k\geq 2$ this bound is a slight improvement on the
upper bound given in Proposition~\ref{pr:boundd}.

\begin{lemma} 
\label{mu-bound-k}
Let $G=(V,E)$ be a connected graph with  $n$ vertices and maximum degree $\Delta$.
Fix an integer $k\in \{ 2,\ldots, n-1\}$ and let $\mathscr{S}$ be any 
$k$-block system for $G$.
Then 
\[ \mu^+ = \mu^+(\mathscr{S}) \leq \Delta - 1 + \nfrac{1}{k}.
\]
\end{lemma}

\begin{proof}
Fix $X\in\Omega = [q]^V$ and $v\in V$.
Given a configuration $c\in [q]^{S_v}$, let $P$ be
the partition of $S_v$ defined by the nonempty colour classes of $c$.
Define $F \subseteq P$ to be the set of colour classes of $c$ which
correspond to a colour which does not appear on $X|_{\partial S_v}$.
 
Let
\[ A_F = \bigcup_{A \in F}\, A, \qquad 
    A_{\overline{F}} = \bigcup_{A \not\in F}\, A\]
and define $a_F = |A_F|$ and $a_{\overline{F}} = |A_{\overline{F}}|$. Writing
$\mu_{X,v}=\mu_{X,S_v}$ for ease of notation, we have
\begin{align}
\mu_{X,v}(c) &\leq  \left( \sum_{A\in F} |\vol(A,A) | \right) +  
   \left(\sum_{A \not \in F} | \vol(A,\, A \cup \partial S_v) |\right) \nonumber \\
&\leq \left(\sum_{A \in F} |\vol(A,A) |\right) + 
   |\vol(A_{\overline{F}}, A_{\overline{F}} \cup \partial S_v)| \label{eq:d1}. 
\end{align}
Observe that 
\begin{equation}
|\vol( A_{\overline{F}}, A_{\overline{F}}\cup \partial S_v ) |
\leq \Delta a_{\overline{F}} - |\vol(A_{\overline{F}},S_v) |
\leq (\Delta - 1) a_{\overline{F}} + \delta_{F,\emptyset}, \label{eq:d2}
\end{equation}
where the last inequality follows by Proposition~\ref{pr:vol} and noting that 
$\delta_{A_{\overline{F}},S_v} = \delta_{F, \emptyset}$.

Next we claim that for $A\in P$ we have 
\begin{equation}
\label{clear}
|\vol(A,A)| \leq (|A|-1)(\Delta - 1).
\end{equation} 
To ease notation, write $a=|A|$.
If $a=1,2$ then (\ref{clear}) clearly holds
(noting that $\Delta \geq 2$ since $G$ is connected). Next, (\ref{clear})
 holds for $\Delta = 2$ since we have $|\vol(A,A)| \leq a - 1$, where the 
``$-1$'' appears because there is at least one edge leaving $A$ (since $G$ is connected). 
If $a=3$ and $\Delta \geq 3$ then  $|\vol(A,A)| \leq 3$ and 
$(a-1)(\Delta - 1) \geq 4$, so (\ref{clear}) holds. For $a \geq 4$ and 
$\Delta \geq 3$,  we note that $|\vol(A,A)| \leq \Delta a /2$ and check that 
$\Delta a /2 \leq (a-1)(\Delta-1)$ holds in this case.
This proves the claim, establishing (\ref{clear}).
 
Therefore
\begin{equation}
\sum_{A \in F} |\vol(A,A)| \leq 
\sum_{A \in F}(|A| - 1)(\Delta - 1) = (a_F - |F|)(\Delta - 1). 
\label{eq:d3}
\end{equation}
Combining (\ref{eq:d1}), (\ref{eq:d2}) , and (\ref{eq:d3}), we have
\begin{align*}
\mu_{X,v}(c) &\leq (a_F - |F|)(\Delta - 1) + 
   (\Delta - 1)a_{\overline{F}} + \delta_{F,\emptyset} \\
    &= (\Delta - 1)(k - |F|) + \delta_{F,\emptyset}.
\end{align*}
Assuming that $|F|\neq k$, dividing by $k - |F|$ gives the ratio $\Delta-1$ if $F\neq\emptyset$
and gives $\Delta-1+ k^{-1}$ if $F=\emptyset$.
This completes the proof.
\end{proof}

Substituting (\ref{partial-plus-k}), (\ref{psi-k}) and the result of
Lemma~\ref{mu-bound-k} into Theorem~\ref{le:block} gives the following,
noting that $\psi_{\min}\geq \nfrac{1}{n}$.

\begin{theorem}
\label{pr:3items}
Let $G=(V,E)$ be a connected graph with $n$ vertices and
maximum degree $\Delta$.
Fix an integer $k\in \{ 2,\ldots, n-1\}$ 
and let $\mathscr{S}$ be a $k$-block system for
$G$.  Let $\psi$ be the uniform distribution on $\mathscr{S}$. 
Fix $\lambda > 1$.
If 
\[ q \geq 2^{k+1}\, \Delta^{2k}\,  k^{2k+1}\, \lambda^{\Delta - 1 + k^{-1}}\]
then $\tau(\mathcal{M}_{\mathrm{BD}}, \varepsilon) \leq 2n\log(n \varepsilon^{-1})$.
\end{theorem}

To further illustrate the use of Theorem~\ref{le:block} we apply
it to the grid.  Although our results are not as sharp as those
discussed in \cite{ull}, using the structure of the grid we are able to prove
an upper bound on $\mu^+$ which is close to the lower bound
given in Proposition~\ref{pr:boundd}. (See Lemma~\ref{mu-bound-grid} below.)

For convenience, rather than considering the $L \times L$ 
grid, we consider the toroidal $L$-grid $G=(V,E)$, where $V = (\field{Z}/L \field{Z})^2$, and $(a,b)(c,d) \in E$ if and only if, in $\field{Z}/L\field{Z}$, 
\[ \text{either} \quad (a-c = \pm 1 \, \text{ and } \, b-d=0) 
\quad \text{or} \quad (b-d = \pm 1 \, \text{ and } \, a-c=0).\] 
Note that the toroidal $L$-grid has $n:=L^2$ vertices. The arguments below can be adapted to higher dimensions and to graphs with different grid topologies provided that the graph is locally a grid.

Let $\mathscr{S}$ be the set of all $r\times r$ subgrids of $G$,
where $r\leq L-2$.  Then $\mathscr{S}$ is a $r^2$-block system.
Let $\psi$ be the uniform distribution on $\mathscr{S}$.
To apply Theorem~\ref{le:block} we must calculate upper bounds on the
parameters.

Firstly, note that 
\begin{equation}
\label{partial-plus-grid}
 \partial^+ = (4r)^{4r}
\end{equation}
since $|\partial S| = 4r$ for all $S\in\mathscr{S}$.
Next, for $v\in V$ we have
$\psi(v) = r^2/L^2$ and $\psi_{\partial}(v) = 4r/L^2$,
and so
\begin{equation}
\label{psi-grid}
\Psi = \frac{4}{r}.
\end{equation}
In order to obtain a tighter bound on $\mu^+$ we need more information
about expansion properties of the grid.
If $U$, $W$ are disjoint sets of vertices, we write $E(U,W)$ for the set of edges
with one endvertex in $U$ and one endvertex in $W$.

\begin{lemma} \label{expansion}
Let $G=(V,E)$ be an $L\times L$ grid and let $S \subseteq V$ be the vertices of an $r\times r$ subgrid. If $T \subseteq S$ and
$|T| = t'$ then $|\vol(T, T)| \leq 2t' - 2 \sqrt{t'}$ and 
$|\vol(T,S \cup \partial S)| \geq 2t' + 2\sqrt{t'}$.
\end{lemma}

\begin{proof}
For $T \subseteq S$, we define 
$\overline{T} = (S \cup \partial S) \setminus T$. 
First, we claim that
\begin{equation}
\label{claim}
  \text{ if } \,\, |E(T,\overline{T})| \leq 4t \,\,\,
   \text{ then } \,\,\, |T| \leq t^2.
\end{equation}
To prove the claim, let us choose $T$ such that $|T|$ is maximised subject to 
$|E(T, \overline{T})| \leq 4t$. We may assume that $G[T]$ is connected or else we can translate components to connect $G[T]$ without increasing $|E(T, \overline{T})|$. Furthermore, we may assume that $T$ is convex (that is, $T$ is a rectangular subgrid) because if $T$ has any ``missing corners'' (that is, a vertex outside $T$ with at least two neighbours in $T$) then we can add the missing vertex without increasing $|E(T,\overline{T})|$. It is also easy to verify that amongst the rectangles with $|E(T,\overline{T})| = 4t$, the
square (with $t^2$ vertices) has the largest area.  This completes the 
proof of the claim.

Now suppose that $|T|=t'$.  Using the contrapositive of (\ref{claim}),
we have
\[  2\, |\vol(T, T)| = 4|T| - |E(T, \overline{T})| \leq 4t' - 4\sqrt{t'},\]
and dividing by two establishes the first statement.  The second
statement follows since
\[ |\vol(T,S \cup \partial S)| = 4|T| - |\vol(T,T)|.\] 
\end{proof}

For the toroidal grid, we may now give an upper bound for the parameter $\mu^+$ which 
is close to the lower bound proved in Proposition~\ref{pr:boundd}.

\begin{lemma}
\label{mu-bound-grid}
Let $G$ be the toroidal $L \times L$-grid, and let $\mathscr{S}$
be the $r^2$-block system consisting of all $r\times r$ subgrids of $G$.
Then
\[ \mu^+ \leq 2 + \nfrac{2}{r}.\]
\end{lemma}

\begin{proof}
For $v \in V$, let $S_v \in \mathscr{S}$ denote the $r \times r$ subgrid in which $v$ is at the ``top left'' corner.
Suppose that $X\in\Omega$ and $v\in V$.  For a given
$c\in [q]^{S_v}$, let $P$ be the corresponding partition
of $S_v$ given by the colour classes of $c$. As usual, 
let $F \subseteq P$ be the set of colour classes corresponding to colours
which do not appear on $X|_{\partial S_v}$.

Recall the notation $A_F$, $A_{\overline{F}}$,
$a_F$ and $a_{\overline{F}}$ introduced in Lemma~\ref{mu-bound-k}.
As in (\ref{eq:d1}) we write $\mu_{X,v}$ for $\mu_{X,S_v}$,
and find that 
\begin{align*}
\mu_{X,v}(c) 
&\leq \left(\sum_{A \in F} |\vol(A,A) |\right)
    + |\vol(A_{\overline{F}}, A_{\overline{F}} \cup \partial S_v)|. 
\end{align*}
Using Lemma~\ref{claim}, we have
\[
\sum_{A \in F} |\vol(A,A) | \, \leq \, 
   \sum_{A\in F}2(|A| - \sqrt{|A|}) \, = \, 
  2a_F - \sum_{A\in F}2\sqrt{|A|} \, \leq \, 2a_F - 2|F|.
\]
In order to bound $|\vol(A_{\overline{F}}, A_{\overline{F}} \cup \partial S_v)|$, observe first that $\vol(S_v,S_v \cup \partial S_v)$ is the disjoint union of $\vol(A_{\overline{F}}, A_{\overline{F}} \cup \partial S_v)$ and $\vol(A_F, S_v \cup \partial S_v)$. Thus
\begin{align*}
 |\vol(A_{\overline{F}}, A_{\overline{F}} \cup \partial S_v)| 
&=  |\vol(S_v,S_v \cup \partial S_v)| - |\vol(A_F, S_v \cup \partial S_v)| \\
&= 2 r^2 + 2r - |\vol(A_F, S_v \cup \partial S_v)| \\
&\leq 2 r^2 + 2r - 2a_F - 2\sqrt{a_F}  \hspace{2cm} \text{by Lemma~\ref{expansion}} \\
&\leq 2 r^2 + 2r - 2a_F - 2\sqrt{|F|}.
\end{align*}
Combining the three inequalities above, we have 
\begin{align*}
\mu_{X,v}(c)  \leq 2(r^2 - |F|) + 2(r - \sqrt{|F|}) 
    &= (r^2 - |F|) \left( 2 + \frac{2}{r + \sqrt{|F|}} \right) \\
   &\leq (r^2 - |F|)\left(2 + \nfrac{2}{r}\right). 
\end{align*}
For all $F$ with $|F|\neq r^2$, dividing by $r^2-|F|$ gives
the value $2+ \nfrac{2}{r}$, completing the proof.
\end{proof}

Substituting (\ref{partial-plus-grid}), (\ref{psi-grid}) and the
result of Lemma~\ref{mu-bound-grid} into Theorem~\ref{le:block}
gives the following, noting that $\psi_{\min} = r^2/L^2$.

\begin{theorem}
\label{th:grid}
Let $G$ be the toroidal $L \times L$-grid (with $n=L^2$ vertices) and let $\mathscr{S}$ be the $r^2$-block system consisting of the set of $r\times r$ subgrids of $G$,
for some $r\leq L-2$.  Given $\lambda > 1$,
 if 
\[ q \geq  2^{r^2 + 8r+3} \, r^{2r^2 + 4r +1}\, \lambda^{2 + \frac{2}{r}}\] 
then for $\mathcal{M}_{\mathrm{BD}} = \mathcal{M}_{\mathrm{BD}}^{\mathscr{S}}(G, \lambda, q)$, we have $\tau(\mathcal{M}_{\mathrm{BD}}, \varepsilon) \leq 2n\log(n \varepsilon^{-1})/r^2$.
\end{theorem}

\subsection{Glauber dynamics via Markov chain comparison}
\label{se:comparison}

The mixing time of two Markov chains on the same state space can be compared using
comparison techniques, building on the work of Diaconis and Saloff-Coste~\cite{DS}.
We now describe the machinery needed to compare the mixing times of the Glauber dynamics and 
the block dynamics. 

Suppose that $\mathcal{M}$ is a reversible, ergodic Markov chain on state space $\Omega$ with transition matrix $P$ and stationary distribution $\pi$. Let $\mathcal{M}'$ be another reversible, ergodic Markov chain on $\Omega$ with transition matrix $P'$ and the same stationary distribution.

We say a transition $(x,y)$ of $\mathcal{M}$ (respectively, $\mathcal{M}'$) is positive if $P(x,y)>0$ (respectively, $P'(x,y)>0$); here we allow the possibility that $x=y$. 
For every positive transition $(x,y)$ of $\mathcal{M}'$, let $\mathcal{P}_{x,y}$ be the set of paths $\gamma = (x=x_0, \ldots, x_k=y)$ such that all the $x_i$ are distinct and each $(x_i,x_{i+1})$ is a positive transition of $\mathcal{M}$. Let $\mathcal{P} = \cup \mathcal{P}_{x,y}$, where the union is taken over all positive transitions $(x,y)$ of $\mathcal{M}'$ with $x\neq y$.  

We write $| \gamma |$ to denote the length of the path $\gamma$ so that, for example, $|\gamma|=k$ for $\gamma = (x_0, \ldots, x_k)$.

An $(\mathcal{M},\mathcal{M}')$-flow is a function $f$ from $\mathcal{P}$ to the interval $[0,1]$ such that for every positive transition $(x,y)$ of $\mathcal{M}'$ with $x\neq y$, we have
\[
 \sum_{\gamma \in \mathcal{P}_{x,y}} f(\gamma) = \pi(x)P'(x,y).
\]
For a positive transition $(z,w)$ of $\mathcal{M}$, the congestion of $(z,w)$ is defined to be
\[
 A_{z,w}(f) = \frac{1}{\pi(z)P(z,w)} \sum_{\gamma \in \mathcal{P}:\, (z,w) \in \gamma } |\gamma| f(\gamma).
\]
The congestion of the flow is defined to be $A(f) = \max A_{z,w}(f)$, where the maximum is taken over all positive transitions $(z,w)$ of $\mathcal{M}$ with $z\neq w$. 

The essence of the comparison technique of Diaconis and Saloff-Coste~\cite{DS} is that the 
the eigenvalues of $\mathcal{M}$ and $\mathcal{M}'$ can be related using the parameter $A(f)$. 
Randall and Tetali~\cite[Theorem 1]{RT} used this result to compare the mixing times of two 
reversible ergodic Markov chains with the same stationary distribution, under the assumption that 
the second-largest eigenvalue
(of the corresponding transition matrices) is larger in absolute value than the
smallest eigenvalue.  (See the discussion above
Theorem~1 of \cite{RT}.)
For convenience, we will use the following theorem, which is obtained from~\cite[Theorem 10]{DGJM}
by specialising to Markov chains with no negative eigenvalues.

\begin{theorem} \emph{\cite[Theorem 10]{DGJM}}\
\label{th:comparison}
Suppose that $\mathcal{M}$ is a reversible ergodic Markov chain with transition matrix $P$ and stationary distribution $\pi$ and that $\mathcal{M}'$ is another reversible ergodic Markov chain with the same stationary distribution. 
Suppose that $f$ is an $(\mathcal{M},\mathcal{M}')$-flow.  If $\mathcal{M}$ has no negative eigenvalues 
then for any $0 < \delta < \frac{1}{2}$, we have
\[
 \tau_x(\mathcal{M}, \varepsilon) \leq 
A(f) \left(  \frac{ \tau(\mathcal{M}', \delta) }{ \log(1/ 2\delta)} + 1 \right) \,
\log \frac{1}{\varepsilon \pi(x)}.
\]
\end{theorem}

Now we apply the above theorem to compare the mixing time of the Glauber dynamics and the block dynamics. 
Write $\tau(\mathcal{M'}) = \tau(\mathcal{M'},\nfrac{1}{2e})$.

\begin{lemma}
\label{le:comparison}
Let $G=(V,E)$ be an $n$-vertex graph of maximum degree $\Delta$. Given $\lambda > 1$, a positive integer $q$, a block system $\mathscr{S}$ for $G$ with
maximum block size $s$, and $\psi$ a probability distribution on 
$\mathscr{S}$,  write $\mathcal{M} = \mathcal{M}_{\mathrm{GD}}(G, \lambda, q)$ and $\mathcal{M}' = \mathcal{M}_{\mathrm{BD}}^{\mathscr{S}, \psi}(G, \lambda,q)$.
Then for all $\varepsilon >0$
we have
\[
\tau(\mathcal{M}, \varepsilon)  \leq  2s\, q^{s+1}\, \lambda^{\Delta (s+1)} \,
   \tau(\mathcal{M}')\, n\,
   \left(n\log{ (q \lambda^{\Delta /2})} + \log( \varepsilon^{-1})\right).
\] 
\end{lemma}

\begin{proof}
As before, let $P$ and $P'$ be the transition matrices of $\mathcal{M}$ and $\mathcal{M}'$ respectively.
We note at the outset that both $\mathcal{M}$ and $\mathcal{M}'$ have the Gibbs distribution $\pi$ as their stationary distribution.  It is proved in~\cite[Section 2.1]{DGU} 
that the Glauber dynamics $\mathcal{M}$ has no negative
eigenvalues,
so we may apply Theorem~\ref{th:comparison}.

We construct an $(\mathcal{M},\mathcal{M}')$-flow and analyse its congestion. Recall that a transition in $\mathcal{M}'$ is obtained by starting at some $X \in \Omega = [q]^V$, selecting $S \in \mathscr{S}$ at random using the distribution $\psi$ and then updating the configuration of $S$ to some configuration $c \in [q]^S$ chosen randomly using the distribution $\phi=\phi_{X,S}$. The resulting
configuration is denoted by $X^{(S,c)}$. Let $h(X,S,c):=\psi(S)\phi_{X,S}(c)$ be the probability that this pair $(S,c)$ is chosen. In particular, if $(X,Y)$ is a transition of $\mathcal{M}'$ then
\[
 P'(X,Y) = \sum_{(S,c):\,  Y=c \cup X^{(S,c)}}\, h(X,S,c),
\]
Fix an ordering of the vertices of $G$.  For each $X \in \Omega$, $S \in \mathscr{S}$, and a configuration $c \in [q]^S$ of $S$, we define the path $\gamma(X,S,c)$ from $X$ to $X^{(S,c)}$ as follows: starting from $X$, consider each vertex $v \in \{u \in S: X(u) \not= c(u) \}$, one at a time and in increasing vertex order, and  change the colour of $v$ from $X(v)$ to $c(v)$. Thus 
$\gamma(X,S,c)$ is a path in $\Omega$ from $X$ to $X^{(S,c)}$ using positive transitions of $\mathcal{M}$. 

We define an $(\mathcal{M}, \mathcal{M}')$-flow $f$ by setting $f(\gamma(X,S,c))=\pi(X)h(X,S,c)$ for all $(X,S,c)$ and $f(\gamma) = 0$ for all other paths $\gamma$. To verify that this is indeed an $(\mathcal{M},\mathcal{M}')$-flow, given a positive transition $(X,Y)$ of $\mathcal{M}'$ with $X\neq Y$, we have
\[
\sum_{\gamma \in \mathcal{P}_{X,Y}}f(\gamma) = 
\sum_{\gamma = \gamma(X,S,c):\,  Y= X^{(S,c)} }\, f(\gamma)
= \sum_{(S,c):\, Y = X^{(S,c)}}
   \pi(X)\, h(X,S,c) = \pi(X)\, P'(X,Y).
\]

Next we bound the congestion of this flow. Let $(Z,W)$ be a positive transition of $\mathcal{M}$ with $Z \not= W$.
Then the configurations $Z$ and $W$ differ on only one vertex, say $v$. The path $\gamma(X,S,c)$ uses the transition $(Z,W)$ only if $v \in S$ and the configurations $X$ and $Z$ differ on a subset of $S$. Thus we have  
\begin{align*}
 A_{Z,W}(f) 
&=   \frac{1}{\pi(Z)P(Z,W)} \, \sum_{\gamma \in \mathcal{P}:\, (Z,W) \in \gamma }
  \,\,  |\gamma| \, f(\gamma) \\
&\leq  \frac{1}{\pi(Z)P(Z,W)} \, \sum_{S:\, v\in S}\,\,
     \sum_{X: \,  
        X|_{\overline{S}} = Z|_{\overline{S}} }\,\, \sum_{c\in [q]^S}\,
          |S| \cdot f(\gamma(X,S,c)) \\
&\leq  \frac{s}{\pi(Z)P(Z,W)} \, \sum_{S:\, v\in S}\,\, \sum_{X: \,  
   X|_{\overline{S}} = Z|_{\overline{S}} } \,\, \sum_{c\in [q]^S} \, 
     \pi(X)\, h(X,S,c) \\
&\leq  \frac{s}{P(Z,W)} \, \sum_{S: \, v \in S}\, \, \sum_{X:
    X|_{\overline{S}} = Z|_{\overline{S}} }\, \frac{\pi(X)}{\pi(Z)} \, \psi(S).
\end{align*}
If $X$ and $Z$ differ on at most $s$ vertices, and hence on at most
$\Delta s$ edges,  then 
\[ \frac{\pi(X)}{\pi(Z)} \leq \lambda^{\Delta s}.\]  
Also, for any positive transition $(Z,W)$ of $\mathcal{M}$ we have
\[
P(Z,W)^{-1} \leq q \lambda^{\Delta}\, n.
\]
Substituting these upper bounds gives
\[ A_{Z,W}(f) \leq \, s q \lambda^\Delta n  \sum_{S: \, v \in S}\,\psi(S) \,
       \sum_{X:\, X|_{\overline{S}} = Z|_{\overline{S}} } \,
              \lambda^{\Delta s} \,
\leq \, sq \lambda^{\Delta} q^s \lambda^{\Delta s} \psi(v) n \,
\leq \, s q^{s+1} \lambda^{\Delta ( s + 1)} n,
\]
since $\psi(v)\leq 1$.
We conclude that $A(f) \leq s q^{s+1} \lambda^{\Delta (s+1)} n$.

Now apply Theorem~\ref{th:comparison} with $\delta = 1/(2e)$.  
For all $Z \in \Omega$, we have the crude bound 
\[
\pi(Z) \geq (q^n\, \lambda^{m})^{-1} \geq (q^n\, \lambda^{\Delta n/2})^{-1},
\]
which leads to
\begin{align*}
 \tau(\mathcal{M}, \varepsilon) &\leq 
  sq^{s+1} \lambda^{\Delta (s+1)}n 
\left(  \tau(\mathcal{M}') + 1 \right) \,
\log{ (q^n \lambda^{\Delta n /2} \varepsilon^{-1})} \\
&\leq  2s\, q^{s+1} \lambda^{\Delta (s+1)}\,  n\,
   \tau(\mathcal{M}')\,
  \left(n\log(q \lambda^{\Delta /2}) + \log(\varepsilon^{-1})\right), 
\end{align*}
as claimed.
\end{proof}

We would expect that the mixing time for Glauber dynamics should decrease as $q$ increases, but the bound given in Lemma~\ref{le:comparison} becomes worse for larger values of $q$. However, by combining Lemma~\ref{le:comparison} with Proposition~\ref{le:vertex}, we can avoid this problem.

\begin{corollary}
\label{co:comparison}
Let $G=(V,E)$ be an $n$-vertex graph of maximum degree $\Delta$. Given $\lambda > 1$, a positive integer $q$, a block system $\mathscr{S}$ for $G$
with maximum block size $s$, and $\psi$ a probability distribution on 
$\mathscr{S}$, write $\mathcal{M} = \mathcal{M}_{\mathrm{GD}}(G, \lambda, q)$ and 
$\mathcal{M}' = \mathcal{M}_{\mathrm{BD}}^{\mathscr{S}, \psi}(G, \lambda,q)$. Then for $\varepsilon > 0$ we have
\[
\tau(\mathcal{M}, \varepsilon)  \leq  \begin{cases}
2s (\Delta\lambda^{2\Delta})^{s+1}\, \tau(\mathcal{M}')\, n \left( 
  n\log(\Delta\lambda^{3\Delta/2}) + \log(\varepsilon^{-1})\right) 
 & \text{ if $q < \Delta \lambda^\Delta + 1$,}\\
  (\Delta+1) n\log(n\varepsilon^{-1}) & \text{ if $q \geq \Delta\lambda^\Delta + 1$}.
\end{cases}
\] 
\end{corollary}

\begin{proof}
If $q < \Delta \lambda^{\Delta} + 1$ then the corollary holds by 
Lemma~\ref{le:comparison}, while if $q \geq \Delta \lambda^{\Delta} + 1$ 
then the corollary holds by Proposition~\ref{le:vertex}.
\end{proof}

We complete this section by applying the previous corollary to the block dynamics results obtained in the previous subsection to obtain rapid mixing results for Glauber dynamics.

\begin{theorem}
Let $G=(V,E)$ be an $n$-vertex connected graph with maximum degree $\Delta$, and fix $\lambda > 1$. For every positive integer 
$k \leq n$, if $q \geq 2^{k+1}\Delta^{2k} k^{2k+1} \lambda^{\Delta - 1 + k^{-1}}$ then for $\mathcal{M}_{\mathrm{GD}} = \mathcal{M}_{\mathrm{GD}}(G, \lambda, q)$, we have 
\[
\tau(\mathcal{M}_{\mathrm{GD}}, \varepsilon)  \leq  
4k\, (\Delta\, \lambda^{2\Delta})^{k+1}\, 
n^2\log(2en)\,\left(n\log{ (\Delta \lambda^{3\Delta/2})} + \log(\varepsilon^{-1})
   \right).   
\]
\label{th:glaubercompare}
\end{theorem}
\begin{proof}
Take an arbitrary $k$-block system $\mathscr{S}$ for $G$, and let $\psi$ be the uniform distribution on
$\mathscr{S}$.  Theorem~\ref{pr:3items} provides
a bound on the mixing time of the block dynamics with respect to $\mathscr{S}$. 
Then apply Corollary~\ref{co:comparison} to this bound.

Here any $k$-block system $\mathscr{S}$ may be used 
(recall the definition after the proof of Proposition~\ref{pr:boundd}). For any connected graph $G=(V,E)$, one can easily obtain a $k$-block system $\mathscr{S} = \{S_v: v \in V\}$ by taking $S_v$ to be the first $k$ vertices in any breadth-first search starting at $v$. 
\end{proof}

\begin{theorem}
\label{th:glaubergrid}
Let $G=(V,E)$ be the toroidal $L \times L$-grid (with $n=L^2$ vertices), and fix $\lambda > 1$. For every positive 
integer $r \leq L-2$, if 
$q \geq 2^{r^2+8r+3}\, r^{2r^2 + 4r+1}\, \lambda^{2 + \frac{2}{r}}$ then for $\mathcal{M}_{\mathrm{GD}} = \mathcal{M}_{\mathrm{GD}}(G, \lambda, q)$, we have 
\[
\tau(\mathcal{M}_{\mathrm{GD}}, \varepsilon)  \leq  
 4 \, (4 \lambda^{8})^{r^2+1}\, n^2\, \log(2en)\, \left(n\log(4 \lambda^6) +
\log( \varepsilon^{-1}) \right).
\]
\end{theorem}
\begin{proof}
We apply Corollary~\ref{co:comparison} to the mixing time of the block dynamics in Theorem~\ref{th:grid}.
(Recall that the block system used is the set of $r \times r$ subgrids.)
\end{proof}

\section{An extremal problem}
\label{se:extremal}

In this section, we investigate how large the partition function of a bounded-degree graph can be. We require this result in the next section, where we give bounds on the number of colours below which Glauber dynamics mixes slowly, although the result may be of independent interest.

In this section, we allow graphs to have multiple edges, but not loops.
For fixed numbers $n$ the number of vertices, $m$ the number of edges, $\Delta$ the maximum degree, $\lambda \geq 1$ the activity, and $q$ the number of colours, we define
\[
 Z((n,m, \Delta), \lambda, q) = \max_{G} Z(G, \lambda, q), 
\]
where the maximum is over all graphs $G$ with $n$ vertices, $m$ edges, and maximum degree $\Delta$. 

We now describe the class of graphs that will turn out to be extremal for the above parameter.
Fix positive integers $n$, $m$, and $\Delta$ such that $\Delta$ divides $m$ and 
$m \leq \Delta n /2$.
Let $H(n,m,\Delta)=(V,E)$, where $V$ is a set of $n$ vertices and $E$ is obtained by taking any set of $m/\Delta$ independent edges on $V$ and replacing each edge with $\Delta$ multi-edges. Thus $H(n,m,\Delta)$ has $m$ edges and maximum degree $\Delta$.

The main result of this section is the following.
\begin{theorem}
\label{thm:ext}
If $G$ is an $n$-vertex graph with $m$ edges and maximum degree $\Delta$, and $q \in \field{N}$ and $\lambda \geq 1$ are given, then
\[
Z(G, \lambda, q) \leq \left( 1 + q^{-1}(\lambda^\Delta-1) \right)^{\lceil m/\Delta\rceil} q^n.
\] 
In particular, if $ \Delta$ divides $m$, we have equality above for $G = H(n,m,\Delta)$.
\end{theorem}
This will immediately give us the following corollary.
\begin{corollary}
\label{cor:ext}
Let $n,m,\Delta \in \field{N}$ be fixed. Given a number of colours $q$, and activity $\lambda \geq 1$, we have
\[
Z((n,m,\Delta), \lambda, q) \leq \left(1 + q^{-1}(\lambda^\Delta - 1) \right)^{\lceil m/\Delta\rceil} q^n.
\]
\end{corollary}

We begin by giving a brief outline of the proof. Given an $n$-vertex multigraph $G=(V,E)$, and a uniformly random
configuration $\sigma$ of $V$ (i.e. $\sigma$ is a uniformly random element of $[q]^V$),
 let $X$ be the number of monochromatic edges of $G$ in $\sigma$. Observe that  $Z(G, \lambda, q) = \field{E}( \lambda^X )q^n$.
 We proceed by decomposing the edges of $G$ into $\Delta$ forests with $\lceil m/\Delta\rceil$ or $\lfloor m/\Delta\rfloor$ edges each.
 Then we establish that the number of monochromatic edges in a forest with $m'$ edges is distributed as $X \sim \bin(m', q^{-1})$. 
 This allows us to obtain a bound on $\field{E}(\lambda^X)$ and hence prove Theorem~\ref{thm:ext}.  
\begin{lemma}
\label{le:forest}
Let $G=(V,E)$ be a multigraph with $n$ vertices, $m$ edges, and maximum degree $\Delta$. We can find $\Delta$ spanning 
forests $F_1,\ldots, F_\Delta$ on the vertex set $V$ such that each $F_i$ has  $\lceil m/\Delta\rceil$ or $\lfloor m/\Delta\rfloor$  edges and the edges of $F_1,\ldots, F_\Delta$ form a partition of $E$.  
\end{lemma}
\begin{proof}
Recall that the \emph{size} of a graph is the number of edges in the graph.  
We begin by disregarding the condition that the forests should have almost equal size,
 and decompose (the edge set of) $G$ into (the edge sets) of $\Delta$ spanning forests, as follows.
(This follows from \cite{NW}, but for completeness we give a brief proof.) Let $G_1:=G$. Iteratively define $F_i$ to be a spanning forest of $G_i$ of maximum size, and let $G_{i+1}$ be obtained from $G_i$ by deleting the edges of $F_i$. 
By removing the edges of $F_i$ from 
$G_i$, we reduce the degree of every non-isolated vertex in $G_i$ by at least one, and so, in 
particular, we reduce the maximum degree of $G_i$ by at least one. Thus $G_r$ is the empty graph
for some $r \leq \Delta$, 
 giving a decomposition of (the edge set of) $G$ into (the edge sets of) $\Delta$ spanning forests, $F_1, \ldots, F_\Delta$ (some of which may have no edges).  

We denote the size of $F_i$ by $|F_i|$. 
Observe that if $|F_i| > |F_j|+1$ then $F_i$ has fewer components than $F_j$ (since all the forests are spanning), so $F_i$ has at least one edge that connects two components of $F_j$. Removing this edge from $F_i$ and adding it to $F_j$ keeps both $F_i$ and $F_j$ acyclic, but reduces the imbalance in their sizes. Iteratively applying this operation to any pair of forests whose sizes differ by at least two eventually results in all forests having size $\lceil m/\Delta\rceil$ or $\lfloor m/\Delta\rfloor$.
\end{proof}

\begin{lemma}
\label{le:bin}
Let $F=(V,E)$ be a forest and let $\sigma$ be a uniformly random configuration of $V$ (i.e. $\sigma$ is a uniformly random element of $[q]^V$). Let $X$ be the number of monochromatic edges of $F$.  Then $X \sim \bin(m, q^{-1})$, where $m$ is the number of edges in $F$.
\end{lemma}

\begin{proof}
It is sufficient to consider the case when $F$ is a tree. For if not, then we can consider the components of $F$ independently, and use the fact that the sum of $t$ independent binomial random variables of the form $\bin(m_j,p)$ is a binomial random variable $\bin(m_1+\cdots +m_t,p)$.

Now assume that $F$ is a tree, and root $F$ at a vertex $v_0$. Let $v_0, \ldots, v_{n-1}$ be any ordering of the vertices in $V$ such that for every $i$, the parent of $v_i$ is a member of $\{v_1,  \ldots, v_{i-1}\}$. We generate a uniformly random configuration of $V$ by colouring each vertex with a uniformly random colour from $[q]$, independently, in the specified order. Each vertex has probability $1/q$ of being given the same colour as its parent, independently of all previous choices,
and hence each edge has probability $1/q$ of being monochromatic, independently of all previous choices.
 Therefore the total number of monochromatic edges satisfies $X \sim \bin(m, q^{-1})$.
\end{proof}

We will also need the following result, which follows from a generalization of H{\" o}lder's inequality. 

\begin{lemma}
\label{le:dom}
Let $(X_1, \ldots, X_d)$ be a random,  
$\field{R}^d$-valued vector, and suppose there exists a random variable $X$ such that $X_i \sim X$ for all $i = 1, \ldots, d$. Then for all $\lambda > 0$
we have
\[
 \field{E}(\lambda^{X_1 + \cdots + X_d}) \leq \field{E}(\lambda^{dX}).
\] 
\end{lemma}
\begin{proof}
Let $Z_j = \lambda^{X_j}$ and $p_j = d$ for $j=1,\ldots, d$.
Then the result follows from 
the generalised H{\" o}lder's inequality, which states
that
\[ \E\left(\prod_{j=1}^d |Z_j|\right) \leq \prod_{j=1}^d \left(\E |Z_j|^{p_j}\right)^{1/p_j}
\]
for any random variables $Z_1,\ldots, Z_d$ 
and any $p_j\geq 1$ such that $\sum_{j=1}^d 1/p_j = 1$.
(See for example~\cite{Finner}.)
\end{proof}

We are now ready to prove Theorem~\ref{thm:ext}.

\begin{proof}[Proof of Theorem~\ref{thm:ext}]\
By Lemma~\ref{le:forest}, we can decompose the edges of $G$ into $\Delta$ spanning 
forests $F_1, \ldots, F_\Delta$, such that $m_i$, the number of edges in $F_i$, is either $\lceil m/\Delta\rceil$ or $\lfloor m/\Delta\rfloor$. 

Let $\sigma$ be a uniformly random configuration of $V$ (i.e. $\sigma$ is a uniformly random element of $[q]^V$,
and let $X_i$ be the number of monochromatic edges of $F_i$ in the configuration $\sigma$. We know by Lemma~\ref{le:bin} that $X_i \sim \bin(m_i ,q^{-1})$. Then $\mon(\sigma)$, the number of  monochromatic edges of $G$ in $\sigma$, is given by 
$ \mon(\sigma) = X_1 + \cdots + X_\Delta$
and 
\begin{align*}
Z(G, \lambda, q) = q^n \, \field{E}(\lambda^{\mon(\sigma)}) = q^n\, \field{E}(\lambda^{X_1 + \cdots + X_\Delta}) .
\end{align*}

For each $i=1, \ldots, \Delta$, choose $Y_i \sim\bin(\lceil m/\Delta\rceil ,q^{-1})$ such that $\field{P}(Y_i \geq X_i) =1$. Then using the above and Lemma~\ref{le:dom}, we have
\begin{align*}
Z(G, \lambda, q) \leq q^n\, \field{E}(\lambda^{X_1 + \cdots + X_\Delta}) \leq q^n\, \field{E}(\lambda^{Y_1 + \cdots + Y_\Delta}) &\leq q^n\, \field{E}(\lambda^{\Delta Y_1}) \\ &= q^n\, (1 + q^{-1}(\lambda^\Delta-1))^{\lceil m/\Delta \rceil}.
\end{align*}
The last equality holds because $Y_1 \sim \bin(\lceil m/\Delta\rceil,q^{-1})$, so
\[
\field{E}(\lambda^{\Delta Y_1}) 
= \sum_{i=0}^{\lceil m/\Delta \rceil} \binom{\lceil m/\Delta \rceil}{i}\, q^{-i}(1-q^{-1})^{\lceil m/\Delta\rceil -i} \, \lambda^{\Delta i} = (1 + q^{-1}(\lambda^\Delta-1))^{\lceil m/\Delta \rceil}.
\]

Finally, it is easy to check that 
$Z(H(n,m,\Delta), \lambda, q)   = q^n\, (1 + q^{-1}(\lambda^\Delta-1))^{m/\Delta }$ when $\Delta$ divides $m$.
 \end{proof}

\section{Slow mixing}
\label{se:tormix}

We have seen in Section~\ref{se:GD} that for general graphs with maximum degree $\Delta$, the Glauber dynamics mixes rapidly if $q \geq \Delta \lambda^\Delta + 1$.  Some improvements on this were given in Section~\ref{se:comparison}. In this section, we shall see that these general bounds cannot be improved by much (in terms of the exponent of $\lambda$). We give a bound on the number of colours below which Glauber dynamics almost surely mixes slowly for a uniformly random $\Delta$-regular graph.

The technical tool used for most slow-mixing proofs is conductance~\cite{JS89}.
We now introduce the necessary definitions: for convenience, we follow
the treatment given in~\cite{DGJM}. 
Again, $\mathcal{M}$ is a Markov chain with state space $\Omega$, 
transition matrix $P$ and stationary distribution $\pi$.
For $A,B \subseteq \Omega$, define
\[
Q_{\mathcal{M}}(A,B) = \sum_{x \in A,\,  y \in B} \pi(x)P(x,y).
\] 
We define 
\[
\Phi_{\mathcal{M}}(A) = \frac{Q_{\mathcal{M}}(A, \overline{A})}{\pi(A)\pi(\overline{A})},
\]
where $\overline{A}:= \Omega \setminus A$. Finally, we define the 
\emph{conductance} of $\mathcal{M}$ as
\[
 \Phi_{\mathcal{M}} := \min_{A \subseteq \Omega} \Phi_{\mathcal{M}}(A).
\]
We drop the subscript when the Markov chain is clear from the context.
Recall that $\tau(\mathcal{M}) = \tau(\mathcal{M},\nfrac{1}{2e})$.
Conductance gives a lower bound for the mixing time of a Markov chain via the following result.

\begin{theorem} \emph{\cite[Theorem 17]{DGJM}}\,  \label{th:con}
Let $\mathcal{M}$ be an ergodic Markov chain with transition matrix $P$, stationary distribution $\pi$ 
and conductance $\Phi$.  Then
\[ \tau(\mathcal{M}) \geq \frac{e-1}{2e\, \Phi_{\mathcal{M}}}.\]
\end{theorem}

Suppose now that $G=(V,E)$ is an $n$-vertex graph, $\lambda \geq 1$ is given, and $q$ is a number of colours. 
By Theorem~\ref{th:con}, in order to show that $\mathcal{M} = \mathcal{M}_{\mathrm{GD}}(G, \lambda, q)$ mixes slowly, 
it is sufficient to show that its conductance $\Phi_{\mathcal{M}}$ is exponentially small in $n$. 

We will need some more definitions. For $i\in [q]$  and $\sigma \in \Omega$,  define 
\[ \sigma_i = |\{v \in V: \sigma(v) = i\}|.\] 
Next, define the $r$-shell and $r$-ball  around a colour $i$ as follows:
\[
S_r(i) = \{ \sigma: \sigma_i = n-r \},\qquad
B_r(i) = \{ \sigma: \sigma_i \geq n-r \}.
\]
We see that $B_r(i)$ is the set of configurations at distance at most $r$ from the all-$i$ configuration, and $S_r(i)$ is the set of configurations at distance exactly $r$ from the all-$i$ configuration. To simplify notation, we write $B_r = B_r(1)$ and $S_r=S_r(1)$ for the $r$-ball and $r$-shell around colour 1.

For an $n$-vertex graph $G=(V,E)$ and $r$ is a positive integer satisfying $r \leq n/2$, we define
\[
 \alpha_r(G) = \frac{1}{r} \,\max_{\stackrel{S \subseteq V}{|S|=r}}\, e_G(S),
\]
where $e_G(S)$ is the number of edges of $G$ inside $S$.  
This quantity is low when the edge-expansion of $r$-vertex subgraphs of $G$ is high.
We now establish a uniform bound on the conductance of $\mathcal{M}_{\mathrm{GD}}(G,\lambda,q)$
which holds when $\alpha_r(G)$ and $q$ are sufficiently small.

\begin{lemma}
\label{lem:conductance}
Let $\lambda \geq 1$ and let $\Delta\geq 2$ be an integer.   
Fix $\kappa\in \big(1,\nfrac{\Delta}{2}\big]$ and let $\beta\in (0,1)$.
Suppose that $n\geq \beta^{-1}(2 + \Delta\log_2\lambda)$ is an 
integer and let $r=\lfloor \beta n\rfloor$.
Let $G$ be a $\Delta$-regular, $n$-vertex graph
such that $\alpha_r(G)\leq \kappa$.
Finally, suppose that $q\geq 2$ is an integer which satisfies
\begin{equation}
\label{q-assumption} 
q-1 
\leq \frac{\beta^2}{256\, e^2}\,\,
        \lambda^{\Delta- \kappa - \frac{\kappa^2}{\Delta - \kappa}}.
\end{equation}
Then the conductance of the Markov chain $\mathcal{M}=\mathcal{M}_{\mathrm{GD}}(G,\lambda,q)$ is bounded by
\begin{equation}
\label{eq:conductance}
\Phi_{\mathcal{M}} 
 \leq \frac{2}{\sqrt{2\pi r}}\, 2^{-r}.
\end{equation}
\end{lemma}

\begin{proof}
We bound $\Phi_{\mathcal{M}}$ by estimating $\Phi_{\mathcal{M}}(B_r)$. 
Let $P$ be the transition matrix for $\mathcal{M}$ and let $\pi$ be the stationary distribution of $\mathcal{M}$ (that is, the Gibbs distribution). We have
\begin{align*}
\Phi_{\mathcal{M}}   \leq   \Phi_{\mathcal{M}}(B_r)  = \frac{\sum_{x \in B_r,\, 
  y \in \overline{B_r}} \pi(x)P(x,y)}{\pi(B_r)\pi(\overline{B_r})} 
&=  \frac{\sum_{x \in S_r,\, y \in \overline{B_r}} \pi(x)P(x,y)}{\pi(B_r)\pi(\overline{B_r})} \\
&\leq \frac{\pi(S_r)}{\pi(B_r) \pi(\overline{B_r})} \\
& \leq \frac{2\, \pi(S_r)}{\pi(B_r)},
\end{align*}
where the last inequality follows because $\pi(\overline{B_r}) \geq \frac{1}{2}$ (assuming that $q \geq 2$). 

Let $Z=Z(G,\lambda,q)$ be the partition function and
write $m=\Delta n /2$ for the number of edges in $G$. Now 
$\pi(B_r) \geq Z^{-1} \lambda^m$ since the all-$1$ configuration belongs to $B_r$. Next we obtain a lower bound on $\pi(S_r)$.

Suppose that $A\subseteq V$ with $|A|=r$. Writing $E(A)$ for the set of edges 
of $G$ inside $A$, 
we know that
$|E(A)| \leq \alpha_r(G)r\leq \kappa r$. 
Observe that
$|E(A,\overline{A})| = \Delta r - 2|E(A)|$ because $\Delta r$ counts each edge in $E(A)$ twice.
Hence
\begin{align*}
 |E(\overline{A})| = m - |E(A,\overline{A})| - |E(A)| 
&= m - (\Delta r - 2|E(A)|) - |E(A)|  \\
&= m - \Delta r +|E(A)|\\
& \leq m - (\Delta - \kappa)r. 
\end{align*}
Therefore
\begin{align*}
\pi(S_r) 
= Z^{-1} \sum_{\sigma \in S_r} \lambda^{\mon(\sigma)} 
&= Z^{-1} \sum_{A \subseteq V : |A|=r} \lambda^{|E(\overline{A})|} \cdot Z(G[A],\lambda,q-1) \\
&\leq Z^{-1} \sum_{A \subseteq V : |A|=r} \lambda^{m - (\Delta - \kappa)r} \cdot Z(G[A],\lambda,q-1) \\
&\leq Z^{-1} \binom{n}{r} \lambda^{m - (\Delta - \kappa)r} \cdot 
Z((r,\lceil \kappa r \rceil, \Delta) ,\lambda,q-1).
\end{align*}
The final inequality uses the fact that when $\lambda\geq 1$, the partition function is nondecreasing under the 
addition of edges.
Combining these bounds shows that
\begin{equation}
\label{conductanceZ}  \Phi_{\mathcal{M}} 
 \leq 2\binom{n}{r} \lambda^{- (\Delta - \kappa)r} \cdot 
   Z((r,\lceil\kappa r \rceil,\Delta) ,\lambda,q-1).
\end{equation}
Using Corollary~\ref{cor:ext}, we have
\begin{align*}
Z((r,\lceil \kappa r\rceil,\Delta) ,\lambda,q-1) 
&\leq  (1 + (q-1)^{-1}\lambda^\Delta)^{\lceil \kappa r/\Delta \rceil}(q-1)^r \\
&\leq (2(q-1)^{-1}\lambda^\Delta)^{\lceil \kappa r/\Delta \rceil}(q-1)^r \\
&\leq 2\lambda^{\Delta}(2(q-1)^{-1}\lambda^\Delta)^{ \kappa r/\Delta }(q-1)^r \\
&\leq  \left( 4 \lambda^{\kappa} (q-1)^{\frac{\Delta - \kappa}{\Delta}} \right) ^r.
\end{align*}
Here the second 
inequality uses the fact that $q-1\leq \lambda^\Delta$
(which follows from (\ref{q-assumption})), 
and the final inequality follows since
$\kappa/\Delta \leq \nfrac{1}{2}$  as well as 
the fact that $2^r \geq 2\lambda^{\Delta}$ (by our choice of sufficiently large $n$).  
Substituting this into (\ref{conductanceZ}) and applying
the well-known inequality
\[  \binom{n}{r} \leq \frac{n^r}{r!} \leq 
\frac{1}{\sqrt{2\pi r}}   \left(\frac{en}{r}\right)^r    \]
gives
\[ \Phi_{\mathcal{M}} \leq
   \frac{2}{\sqrt{2\pi r}}\, \left( \frac{4en}{r}\, \lambda^{-\Delta - 2\kappa}\, (q-1)^{(\Delta-\kappa)/\Delta}
            \right)^{r}.
\]
Now raising both sides of (\ref{q-assumption}) to the power $(\Delta-\kappa)/\Delta$
and rearranging shows that
\[
\frac{4en}{r}\, \lambda^{- (\Delta - 2\kappa)} 
      (q-1)^{\frac{\Delta - \kappa}{\Delta}} \leq \frac{\beta n}{4r}\leq \frac12.
\]
Therefore
$\Phi_{\mathcal{M}} \leq \frac{2}{\sqrt{2\pi r}} \, 2^{-r}$,  
as claimed.
\end{proof}

Let $\mathcal{G}_{n,\Delta}$ denote the uniform probability space
of all $\Delta$-regular graphs on the vertex set $[n]=\{ 1,2,\ldots, n\}$,
restricting to $n$ even if $\Delta$ is odd.  That is, ``$G\in\mathcal{G}_{n,\Delta}$''
means that $G$ is a uniformly chosen $\Delta$-regular graph on the vertex set $[n]$.
In a sequence of probability spaces indexed by $n$, an event holds \emph{asymptotically almost
surely} (a.a.s.) if the probability that the event holds tends to 1 as $n\to\infty$.

Next, given $\kappa$ we show how to choose $r$ in order to ensure that with high probability,
a random $\Delta$-regular graph $G$ satisfies $\alpha_r(G)\leq \kappa$.

\begin{lemma} 
\label{le:expander}
Fix $\Delta \geq 3$ and let $\kappa \in \big(1,\nfrac{\Delta}{2}\big]$.  
Let 
\begin{equation} \label{eq:expan}
 \beta = \nfrac{1}{2}\, e^{-\big(1 + \frac{2}{\kappa-1}\big)}\,  
    \bigg( \frac{\Delta}{2\kappa} \bigg)^{ -\big( 1 + \frac{1}{\kappa-1}\big) },
\end{equation}  
and for each positive integer $n\geq \beta^{-1}$, define 
$r=r(n) = \lfloor \beta n\rfloor$, which is a positive integer.
Let $G\in\mathcal{G}_{n,\Delta}$.  Then a.a.s.\  $\alpha_r(G)\leq  \kappa$.
\end{lemma}

\begin{proof}
We use the configuration model of Bollob{\'a}s~\cite{bollobas} to construct 
random regular graphs. In this model, to construct a random $\Delta$-regular 
graph on $n$ vertices, we take $n$ sets (called \emph{buckets}) each 
containing $\Delta$ labelled objects called \emph{points}. Then we take a 
random partition $P$ of the $\Delta n$ points into $\Delta n/2$ 
\emph{pairs}, where each pair is a set of two distinct points.  
We call $P$ a \emph{pairing}. By replacing each bucket by a vertex and 
replacing each pair by an edge between the two corresponding vertices, we 
obtain a multigraph $G(P)$, which may have loops and multiple edges. 
If $G(P)$ is simple then it is $\Delta$-regular.  It has been 
shown~\cite{bollobas} that a random pairing is simple with probability 
tending to 
$\exp{(- \frac{\Delta^2-1}{4})}$ as $n \to \infty$.

Let $m(2a)$ denote the number of pairings of $2a$ points. 
It is well known that
\[  m(2a) = \frac{(2a)!}{a!\, 2^a}.\]
Write $[x]_a = x(x-1)\cdots (x-a+1)$ to denote the
falling factorial.  
Now let $\mathcal{P}_{n,\Delta}$ denote the uniform
probability space on the set of pairings
with $n$ buckets, each containing $\Delta$ points.  Let $B$
be a fixed set of $r$ buckets.  
Given a positive integer $s$, let $m_B(r,s)$ be the
number of pairings in $\mathcal{P}_{n,\Delta}$ in which
at least $s$ pairs are contained in $B$. 
We can obtain an overcount of $m_B(r,s)$ in the following way.
  We first select $s$ pairs within $B$, in 
\[ \frac{ [\Delta r]_{2s}}{s!2^s}\]
ways. Then we pair up the remaining $\Delta n-2s$ points
in $m(\Delta n-2s)$ ways.   
Hence
\[ m_B(r,s) \leq \frac{[\Delta r]_{2s}}{s! 2^s}\, 
    \frac{(\Delta n - 2s)!}{(\Delta n/2 - s)! 2^{\Delta n/2-s}}
   = \frac{(\Delta r)! (\Delta n - 2s)!}
   { 2^{\Delta n/2} s! (\Delta r - 2s)! (\Delta n/2 - s)! }.
\]
Therefore the probability $p(r,s)$ 
that a random pairing in $\mathcal{P}_{n,\Delta}$
has at least $s$ pairs within $B$ is
\[ p(r,s)= \frac{m_B(r,s)}{m(\Delta n)} \leq 
   \binom{\Delta n/2}{s} \frac{ [\Delta r]_{2s}}
           { [\Delta n]_{2s} }
\leq \binom{\Delta n/2}{s} \bigg( \frac{r}{n} \bigg)^{2s}.
\]

Let $X(r,s)$ be the random variable which counts the number of sets
of $r$ buckets which contain at least $s$ pairs of $P$, for 
$P\in\mathcal{P}_{n,\Delta}$. 
Using the inequality $\binom{a}{b} \leq (ea/b)^b$, 
we have
\[
\E(X(r,s)) = \binom{n}{r} p(r,s) 
   \leq \binom{n}{r}\binom{\Delta n/2}{s} \bigg( \frac{r}{n} \bigg)^{2s}
  \leq \left(\frac{en}{r}\right)^r \, \left(\frac{\Delta e r^2}{2sn}
      \right)^s.
\]
Now fix $s = \lceil\kappa r\rceil$ where, recall,  
$r =\lfloor\beta n\rfloor$. 
By definition of $\beta$ we have $\Delta e r < 2\kappa n$, 
and hence
\[
 \E(X(r,\lceil \kappa r\rceil))  
         \leq \bigg( \frac{ne}{r} \,
   \bigg( \frac{\Delta e r }{2 \kappa n}\bigg)^{\kappa}  \bigg)^{r}\
    \leq ((2\kappa)^{-\kappa} \,e^{\kappa+1} \,     
      \Delta^{\kappa} \,\beta^{\kappa-1})^{r}.
\]

When (\ref{eq:expan}) holds, we see that
\[ (2\kappa)^{-\kappa}\, e^{\kappa+1} \, \Delta^\kappa\, \beta^{\kappa-1} 
   \leq 2^{-(\kappa-1)}  \]
and this upper bound is a constant in $(0,1)$ which is independent of $n$.
Since $r\geq \beta n-1$ it follows that $\E(X(r,\lceil \kappa r\rceil)) = o(1)$, and we 
conclude that 
\[ \E(X(r,\lceil \kappa r\rceil) \mid G(P) \text{ is simple}) \leq 
     \frac{\E(X(r,\lceil \kappa r\rceil))}{\Prob(G(P) \text{ is simple})} =o(1).
\]
This shows that when (\ref{eq:expan}) holds, 
a.a.s.\ $G\in\mathcal{G}_{n, \Delta}$ has the property that all
subsets of vertices of size $r$ have fewer than $\kappa r$ edges. 
\end{proof}

Now we can easily show that when $q$ is sufficiently small and $n$ is sufficiently large,
the mixing time of the Glauber dynamics is slow for almost all $\Delta$-regular graphs.

\begin{theorem}
\label{th:torpid-random}
Fix $\Delta\geq 3$  and let $\kappa \in \big(1, \nfrac{\Delta}{2} \big]$. 
Suppose that $\beta$ is defined by (\ref{eq:expan}) and let $q\geq 2$ be
an integer which satisfies (\ref{q-assumption}).
Let $G\in\mathcal{G}_{n,\Delta}$.
Then a.a.s.\
the Glauber dynamics  $\mathcal{M} = \mathcal{M}_{\mathrm{GD}}(G, \lambda, q)$ satisfies
\[
\tau(\mathcal{M}) \geq 2^{\beta n-4}. \]
\end{theorem}

\begin{proof}
For each positive integer $n\geq \beta^{-1}( 2 + \Delta \log_2 \lambda)$, let
$r = r(n) = \lfloor \beta n\rfloor$, which is a positive integer. 
By Lemma~\ref{le:expander} we know that
a.a.s.\ $G\in\mathcal{G}_{n,\Delta}$ satisfies $\alpha_r(G)\leq \kappa$.
Hence a.a.s.\ the conductance of the corresponding Glauber dynamics $\mathcal{M}_{\mathrm{GD}}(G,\lambda,q)$
is bounded above by
\[ \frac{2}{\sqrt{2\pi r}}\, 2^{-r}\]
by Lemma~\ref{lem:conductance}.
Applying Theorem~\ref{th:con} completes the proof.
\end{proof}

We conclude this section by proving Theorem~\ref{th:main2}
and Theorem~\ref{th:main3}.  

\begin{proof}[Proof of Theorem~\ref{th:main2}]
(i) Given $\eta\in (0,1)$, let $k=\lceil \eta^{-1}\rceil$ and define 
$c_1 = k2^{k+1}(\Delta k)^{2k}$.
If $q > c_1\lambda^{\Delta-1+\eta}$
then $q > c_1 \lambda^{\Delta - 1 + 1/k}$, by choice of $k$.
Then the conclusion follows from Theorem~\ref{th:glaubercompare}.

For (ii), given $\eta \in (0,1)$ define $\kappa = 1+\eta/5$. 
Since $\Delta\geq 3$ we have
  $\kappa\in \big(1,\nfrac{6}{5}\big)\subseteq (1,\nfrac{\Delta}{2}]$.
Define 
\[ c_2 = \nfrac{1}{1024} e^{-4(1 + \frac{1}{\kappa - 1})}\, \left(
  \frac{\Delta}{2\kappa}\right)^{-2(1 + \frac{1}{\kappa-1})}.\]
By our choice of $\kappa$ and since $\Delta\geq 3$, 
we have
\begin{align*}
\kappa + \frac{\kappa^2}{\Delta-\kappa}
 & \leq 1 + \frac{1}{\Delta-1} + \eta.
\end{align*}
Therefore, if
\[
 q-1 \leq c_2 \lambda^{\Delta -1 -\frac{1}{\Delta -1} - \eta} \]
then (\ref{q-assumption}) holds,  and
the result follows by applying Theorem~\ref{th:torpid-random}.
\end{proof}

\begin{proof}[Proof of Theorem~\ref{th:main3}]
The first and third statement follow from substituting $\Delta=4$
into Theorem~\ref{th:main2} (i) and (ii), respectively.
(So $c_3$ is obtained by substituting $\Delta=4$ in $c_1$,
and $c_5$ is obtained from $c_2$ similarly.)

For (ii), let $k = \lceil 2\eta^{-1}\rceil$ and define
$c_4 = (8k-1)\, 2^{k^2+8k}\, k^{2k^2 + 4k}$.
If $q > c_4 \lambda^{2+\eta}$ then $q > c_4 \lambda^{2 + 2/k}$,
by definition of $k$.  Then 
Theorem~\ref{th:glaubergrid} applies, completing the proof.
\end{proof}

\subsection*{Acknowledgements} 
  We are grateful to Ostap Hryniv and Gregory
Markowsky for leading us to the generalised H{\" o}lder's inequality
(and to~\cite{Finner}) for Lemma~\ref{le:dom}. We are also grateful to Mario Ullrich for providing feedback on an earlier draft of this paper. We would also like to thank the referees for their helpful comments.

\section*{Appendix}

Suppose that $q,\, \Delta\geq 3$ are integers and that $B$ is a 
real number. 
We prove that the polynomial 
\[
f(x) := (q-1)x^{\Delta} + (2- q -B)x^{\Delta -1} + Bx -1 
\]
has a double root in $(0,1)$ only if $0 < B = \Theta(q^{\frac{1}{\Delta - 1}})$ i.e.\ $\log B = \frac{\log q}{\Delta - 1} + O(1)$.
 Here all asymptotic notation is with respect to $q \to \infty$.

First we note some properties of $f$. Observe that 
$f''(x) = c_1 x^{\Delta -2} + c_2 x^{\Delta - 3}$ for some constants 
$c_1, c_2$. Thus $f''(x)$ has at most one root in $(0,1)$. This implies that $f'(x)$ has at most one turning point in $(0,1)$ and hence at most two roots in $(0,1)$. Thus $f(x)$ has at most two turning points in $(0,1)$.
This together with the fact that $f(0) = -1$ and $f(1) = 0$ implies that if $f$ has a double root in $(0,1)$, it must be the case that $f(x) \leq 0$ for all $x \in [0,1]$. (To see this, consider the graph of $f$ with the constraints deduced above.)

We show that (i) if $0<B = \omega(q^{\frac{1}{\Delta -1}})$ and 
$q$ is sufficiently large, then $f(x)>0$ for some $x \in (0,1)$; 
(ii) if $B\leq 0$ then $f(x) < 0$ for all $x \in (0,1)$; and 
(iii) if $0 < B = o(q^{\frac{1}{\Delta -1}})$ and $q$ is sufficiently large, 
then $f(x) < 0$ for all $x \in (0,1)$.
Thus in all three cases there is no double root of $f$ in $(0,1)$; the only possibility remaining is that $0<B = \Theta(q^{\frac{1}{\Delta - 1}})$.

Splitting the terms in $f$, we have:
\[
f(x) = (q-1)x^{\Delta} - (q-2)x^{\Delta -1} - Bx^{\Delta -1} + Bx - 1.
\]
First suppose that $0 < B = \omega(q^{\frac{1}{\Delta -1}})$. 
Then $f(q^{- \frac{1}{\Delta -1}})$ is dominated by the fourth term above, which is positive.  Hence $f(q^{- \frac{1}{\Delta -1}})>0$ for $q$ sufficiently large, proving (i).

For (ii) and (iii), first observe that for all $x \in (0,1)$, we have
\[ f(x) = (x-1) \left( (q-1)x^{\Delta -1} + 1 + \sum_{i=1}^{\Delta -2}(1-B)x^i \right). \]
If $B\leq 0$ then for all $x\in (0,1)$, the second factor on the right hand side is positive
and the first factor is negative,  establishing (ii).  

For the remainder of the proof, suppose that 
$0 \leq B = o(q^{\frac{1}{\Delta-1}})$.  Using the
above identity and the fact that $B$ is positive, for all
$x\in (0,1)$ we obtain
\begin{align*}
f(x) 
&\leq (x-1)\left( (q-1)x^{\Delta -1} + 1 + \sum_{i=1}^{\Delta -2}(-B)x^i \right) \\
&\leq (x-1)\left((q-1)x^{\Delta -1} + 1 - \Delta B x \right).
\end{align*}
If $x \in (0, q^{-\frac{1}{\Delta-1}}]$ then $\Delta Bx = o(1)$, so $f(x)<0$ (for all sufficiently large $q$). If $x \in [q^{- \frac{1}
{\Delta -1}},1)$ then it is easy to check that $\Delta Bx = 
o((q-1)x^{\Delta -1})$, so $f(x)<0$ (for all sufficiently large $q$). 
Combining these two statements shows that (iii) holds, completing the proof.
\qed
\end{document}